\documentclass[a4paper,11pt]{article}
\usepackage[left=1in, right=1in, bottom=1.25in, top=1.5in]{geometry}

\usepackage{breakcites}

\usepackage{fullpage}
\usepackage{times}
\usepackage{soul}
\usepackage[utf8]{inputenc}
\usepackage[small]{caption}
\usepackage{graphicx}

\usepackage{url}            
\usepackage{booktabs}       
\usepackage{nicefrac}       
\usepackage{microtype}      
\usepackage{amssymb,amsmath,amsfonts,amstext}

\usepackage{amsthm} 
\usepackage{thmtools} 
\usepackage{mathtools} 

\usepackage{xcolor}         
\definecolor{pku-red}{RGB}{139,0,18}
\usepackage[backref, colorlinks,citecolor=blue,linkcolor=pku-red,urlcolor=pku-red,bookmarks=true]{hyperref}
\usepackage[nameinlink]{cleveref}

\usepackage[vlined,ruled,linesnumbered]{algorithm2e} 
\SetCommentSty{small}

\usepackage{enumitem}

\usepackage{libertine}

\usepackage{tikz}  

\usepackage{bbm}

\newtheorem{theorem}{Theorem}[section]
\newtheorem{corollary}[theorem]{Corollary}
\newtheorem{proposition}[theorem]{Proposition}
\newtheorem{lemma}[theorem]{Lemma}

\theoremstyle{definition}
\newtheorem{definition}[theorem]{Definition}

\newtheorem{condition}{Condition}

\usepackage{natbib}%
\setcitestyle{sort, authoryear,round,citesep={;},aysep={,},yysep={;}}

\usepackage{mathtools}
\usepackage{tcolorbox}

\usepackage{subcaption}

\providecommand{\email}[1]{\href{mailto:#1}{\nolinkurl{#1}\xspace}}

\newcommand\child[1]{\operatorname{Child}(#1)}

\DeclareMathOperator*{\argmax}{arg\,max}
\DeclareMathOperator*{\argmin}{arg\,min}

\DeclareMathOperator{\BR}{BR}

\tikzset{
box/.style ={
rectangle, 
rounded corners =1pt, 
minimum width =50pt, 
minimum height =20pt, 
inner sep=5pt, 
draw=black, 
align = center
}}

\title{
Optimal Private Payoff Manipulation against Commitment in Extensive-form Games
}

\author{Yurong Chen \\Peking University \\\email{chenyurong@pku.edu.cn} \and  Xiaotie Deng \\Peking University \\ \email{xiaotie@pku.edu.cn} \and Yuhao Li \\Columbia University \\ \email{yuhaoli@cs.columbia.edu}}

\date{}

\begin{document}

\maketitle

\begin{abstract}
To take advantage of strategy commitment, a useful tactic of playing games, a leader must learn enough information about the follower's payoff function. 
However, this leaves the follower a chance to
provide fake information and 
influence 
the final game outcome. 
Through a carefully contrived payoff function misreported to the learning leader, the follower may induce an outcome that benefits him more, compared to the ones when he truthfully behaves.

We study the follower's optimal manipulation via such strategic behaviors in extensive-form games.
Followers' different attitudes are taken into account. 
An optimistic follower maximizes his true utility among all game outcomes that can be induced by some payoff function.  A pessimistic follower only considers misreporting payoff functions that induce a unique game outcome. For all the settings considered in this paper, we characterize all the possible game outcomes that can be induced successfully. We show that it is polynomial-time tractable for the follower to find the optimal way of misreporting  his private payoff information. Our work completely resolves this follower's optimal manipulation problem on an extensive-form game tree. 
\end{abstract}

\newpage
\setcounter{page}{1}

\section{Introduction}

In extensive-form games, game trees capture the sequential interaction among players. While the interaction structure of a game is predetermined, a player may still be able to commit to a strategy first at every turn of her move rather than choose strategies simultaneously. This makes her the ``conceptual leader'' and potentially gain more utility. 
Such an idea is originally captured in \cite{von1934marktform}. With the prediction that the follower will best-respond to the commitment,
the leader can commit to a strategy that optimizes her own utility. 
The pair of a leader's optimal commitment and a follower's corresponding best response is
called strong Stackelberg equilibrium (SSE), with the assumption that the follower breaks ties in favor of the leader to assure its existence
\citep{letchford2010computing}.

Generally, the leader can gain potentially
more in SSE than in Nash Equilibrium (NE), since she can always choose her NE strategy 
for the follower to 
best-respond to, so as to reach NE as a feasible
Stackelberg solution. 
Due to its modelling of conceptually asymmetric roles, 
SSE in extensive-form games has inspired a wide range of applications, such as modelling sequential defenses and attacks in security games~\citep{nguyen2019tackling}, addressing execution uncertainties in airport securities~\citep{delle2014game}, and 
so on~\citep{sinha2018stackelberg, cooper2019stackelberg, durkota2019hardening}. 

However, to achieve this ``first-mover'' advantage, the leader needs to have enough knowledge about the follower's payoff function. 
Various learning methods have been proposed to elicit the follower's payoff information through interaction ~\citep{letchford2009learning, Blum2014learning,nika2016three, roth2016watch,peng2019learning}.
A practical example is big data discriminatory pricing, as a result of learning application to the economics. According to different prices, consumers will best-respond by buying or not buying. 
Then 
platforms can construct user profiles via data collection, and set specialized prices for the same products to different customers
\citep{Mahdawi2018is, chang2021does}. 

Such learning process
then leaves
the follower a chance to misreport his private payoff information.
He 
can strategically
act as possessing an alternative payoff function and induce a new game. 
This may improve his real utility on the SSE of the newly induced game. 
An example is shown in Figure \ref{example: cheat can gain more} where the leader commits to a pure strategy. When the follower reports his real payoff function, he gets a value of $3$ on the SSE of the game, as shown in Figure \ref{example: cheat can gain more-truthful}. When he reports a payoff function shown in Figure \ref{example: cheat can gain more-cheat}, his real utility rises to
$4$ on the SSE of the induced game.

Back to the above practical example, there are evidences that customers are aware of the price discrimination and beginning their manipulation, 
like changing their credit cards~\citep{Mahdawi2018is}. We note that even though the leader is fully rational, she might temporarily fail to realize such strategic behaviors. Such cases may happen when the leader is playing games with huge amounts of followers simultaneously, e.g., games between e-commerce platforms and consumers. And a rational follower would definitely seize
and benefit from this chance.
Actually, as will be further illustrated in the related work, the learning and manipulation of private information has shown increasing interests in both learning and game theory communities.

In this paper, we study how a follower optimally misleads the leader through strategic reporting. Specifically, we model the whole process into two phases: in the first phase, the follower reports a payoff function to the leader, which is a simplification of leader's learning follower's payoff information through interaction; in the second phase, 
the leader
computes an SSE in an extensive-form game with perfect information, according to a predetermined game tree, the leader's publicly-known payoff function and the follower's reported payoff function.

Two 
settings of commitment are 
studied: 
we first consider the setting where the leader only commits to a pure strategy. We then proceed to the more general and complicated setting where the leader commits to a behavioral strategy. 
We will regard a probability distribution over leaf nodes of the game tree as the game outcome that the follower aims to finally induce through strategic reporting, because of that 1) different strategy profiles may lead to the same 
distribution; and that 2) it is the distribution that plays a decisive role in calculating the players' expected utilities.

\begin{figure}[!htbp]
\captionsetup[subfigure]{}
\subcaptionbox{When the follower acts according to his true function $U_F$. \label{example: cheat can gain more-truthful}}[.49\textwidth]{%
\centering
\begin{tikzpicture}[scale=0.95, every node/.style={transform shape}]
\node[above left] at (3,1.6) {$F$};
\node[above left] at (1.5,0.8) {$L$};
\node[above right] at (4.5,0.8) {$L$};
\node at (0,-0.3) {$U_L$};
\node at (0,-0.7) {$U_F$};
\node at (0.75,-0.3) {$8$};
\node at (0.75,-0.7) {$1$};
\node at (2.25,-0.3) {$4$};
\node at (2.25,-0.7) {$3$};
\node at (3.75,-0.3) {$2$};
\node at (3.75,-0.7) {$4$};
\node at (5.25,-0.3) {$1$};
\node at (5.25,-0.7) {$2$};
\filldraw (3,1.6) circle (.08)
          (1.5,0.8) circle (.08)
          (4.5,0.8) circle (.08)
          (0.75,0) circle (.08)
          (3.75,0) circle (.08)
          (5.25,0) circle (.08);
\filldraw[blue] (2.25,0) circle (.08);
\draw[thick] (3,1.6)--(1.5,0.8)
             (3,1.6)--(4.5,0.8)
             (1.5,0.8)--(0.75,0)
             (1.5,0.8)--(2.25,0)
             (4.5,0.8)--(3.75,0)
             (4.5,0.8)--(5.25,0);
\draw[->, thick, blue] (2.8,1.58)--(1.6,0.94);
\draw[->, thick, blue] (1.645,0.75)--(2.245,0.11);
\draw[->, thick, blue] (4.645,0.75)--(5.245,0.11);
\end{tikzpicture}
}
\hfill
\subcaptionbox{When the follower acts according to $\tilde{U}_F$.\label{example: cheat can gain more-cheat}}[.49\textwidth]{
\begin{tikzpicture}
\centering
\node[above left] at (3,1.6) {$F$};
\node[above left] at (1.5,0.8) {$L$};
\node[above right] at (4.5,0.8) {$L$};
\node at (0,-0.3) {$U_L$};
\node at (0,-0.7) {$U_F$};
\node at (0.75,-0.3) {$8$};
\node at (0.75,-0.7) {$-9$};
\node at (2.25,-0.3) {$4$};
\node at (2.25,-0.7) {$-9$};
\node at (3.75,-0.3) {$2$};
\node at (3.75,-0.7) {$0$};
\node at (5.25,-0.3) {$1$};
\node at (5.25,-0.7) {$1$};
\filldraw (3,1.6) circle (.08)
          (1.5,0.8) circle (.08)
          (4.5,0.8) circle (.08)
          (0.75,0) circle (.08)
          (2.25,0) circle (.08)
          (5.25,0) circle (.08);
\filldraw[blue] (3.75,0) circle (.08);
\draw[thick] (3,1.6)--(1.5,0.8)
             (3,1.6)--(4.5,0.8)
             (1.5,0.8)--(0.75,0)
             (1.5,0.8)--(2.25,0)
             (4.5,0.8)--(3.75,0)
             (4.5,0.8)--(5.25,0);
\draw[->, thick, blue] (3.2,1.58)--(4.4,0.94);
\draw[->, thick, blue] (1.645,0.75)--(2.245,0.11);
\draw[->, thick, blue] (4.355,0.75)--(3.755,0.11);
\end{tikzpicture}}
\caption{An example how the follower can gain more utility by strategically reporting an alternative payoff function, when the leader only commits to pure strategies. 
The follower acts at the root node, while the leader acts at other internal nodes. The first row below the game tree shows the leader's payoff at each leaf node, while the second row shows the follower's. Blue arrows denote the SSEs in the games. Fig.~\ref{example: cheat can gain more-truthful} shows results when the follower reports his true payoff function $U_F$, while Fig.~\ref{example: cheat can gain more-cheat} shows that when he reports a  payoff function $\tilde{U}_F$. He gains $1$ more utility by strategic reporting. }
\label{example: cheat can gain more}
\end{figure}
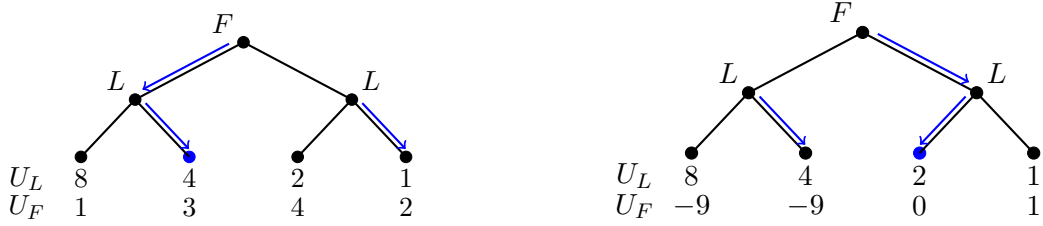
For both settings, we develop four key results as follows: 
\begin{enumerate}
    \item We characterize all the 
    probability distributions  
    that can be led to by \textbf{one} SSE of an induced game, made up of the leader's true payoff function and a follower's proposed payoff function, as  called \textit{inducibility} (\Cref{thm: cheat in pure commit}, \Cref{thm: cheat in behavioral commit}).

    \item We propose polynomial-time algorithms for the follower to find an optimal 
    distribution
    that yields his highest utility among all the inducible ones, 
    and construct a 
    follower's payoff function
    that induces this distribution 
    (\Cref{thm: algorithm for pure case}, \Cref{thm: algorithm for behavioral case}).
    
    \item Ultimately, the bar on the follower's strategic reporting is raised: we characterize all the
    probability distributions
    that there is a way of reporting to make it the \textbf{unique} outcome, led to by all induced SSEs 
    (\Cref{thm: strong inducibility under pure}, \Cref{thm: strong inducibility under behavioral}), as 
    called \textit{strong inducibility}.  We give polynomial-time algorithms to find an optimal distribution for the follower or a near-optimal distribution when the optimal one does not exist (\Cref{thm: algorithm for strongly inducible under pure}, \Cref{thm: algorithm for strongly inducible under behavioral}). 
    
    \item We compare the optimal utilities in one game that a follower can get among different settings, especially those under
    inducible  and strongly inducible distributions, respectively: we fully characterize the games where these two values are equal (\Cref{thm: full characterization of near optimality}).  We say such games satisfy the property \textit{Utility Supremum Equivalence}. 
    
\end{enumerate}

\Cref{relationship graph} summarizes our main results by their relationships with different settings. In more detail, under the inducibility and pure commitment case,
we show that a distribution can be successfully induced so long as the leader's utility is no less than her max-min value. Beyond this setting, max-min value is not enough for such characterization, due to the best responses' subgame perfect properties on extensive-form games. While conditions are becoming stricter and more complicated, we show that the follower can gain no less utility under behavioral commitment case than under pure commitment case. Algorithmically, for the inducibility and pure commitment case,
the time complexities of the above algorithms are optimal.

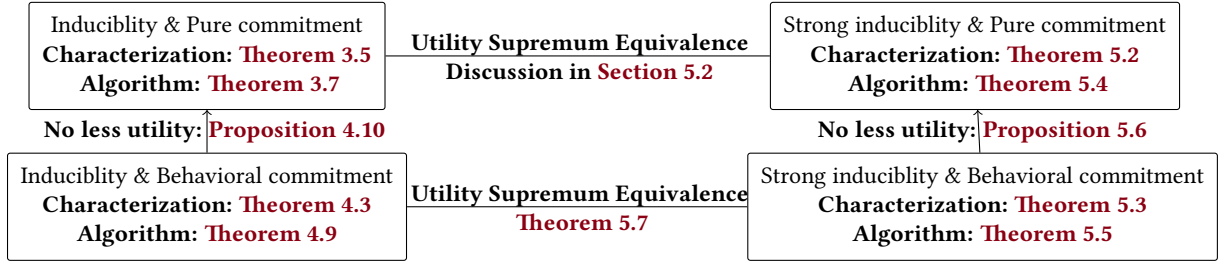
\begin{figure}[!htbp]
\begin{tikzpicture}[scale=1, every node/.style={transform shape}]
\centering
\node[box, font = \footnotesize] (1) at (0,2) {{Induciblity \& Pure commitment}\\{\textbf{Characterization: \Cref{thm: cheat in pure commit}}} \\{ \textbf{Algorithm: \Cref{thm: algorithm for pure case}}}};
\node[box, font = \footnotesize] (2) at (0,0) {{Induciblity \& Behavioral commitment}\\{\textbf{Characterization: \Cref{thm: cheat in behavioral commit}}} \\ {\textbf{Algorithm: \Cref{thm: algorithm for behavioral case}}}};
\node[box, font = \footnotesize] (3) at (10.1,2) {{Strong induciblity \& Pure commitment}\\{\textbf{Characterization: \Cref{thm: strong inducibility under pure}}} \\{ \textbf{Algorithm: \Cref{thm: algorithm for strongly inducible under pure}}}};
\node[box, font = \footnotesize] (4) at (10.2,0) {{Strong induciblity \& Behavioral commitment}\\{\textbf{Characterization: \Cref{thm: strong inducibility under behavioral}}} \\ {\textbf{Algorithm: \Cref{thm: algorithm for strongly inducible under behavioral}}}};
\draw[->] (2)--(1);
\node[auto,sloped,font=\footnotesize] at (0.1,1) {\textbf{No less utility: \Cref{proposition: higher utility}}};
\draw[-] (2)--(4);
\node[auto,sloped,font=\footnotesize, align = center] at (4.9,0) {\textbf{Utility Supremum Equivalence}\\ \textbf{ \Cref{thm: full characterization of near optimality}}};
\draw[->] (4)--(3);
\node[auto,sloped,font=\footnotesize] at (10.2,1) {\textbf{No less utility: \Cref{proposition: higher utility under strongly inducibility}}};
\draw[-] (1)--(3);
\node[auto,sloped,font=\footnotesize, align = center] at (4.9,2) {\textbf{Utility Supremum Equivalence}\\ \textbf{Discussion in \Cref{sec: efficiency of strong inducibility}}};
\end{tikzpicture}
\caption{Summary of our results under different settings. Edges link to the settings where the optimal utilities that the follower can get through strategic reporting are compared. }
\label{relationship graph}
\end{figure}

\subsection*{Technique Overview}
Our technique exploits the sequential structure of extensive-form games, modeled by game trees. 

First, we abstract from strategy profiles and treat 
probability distributions on leaf nodes of game trees
as the game outcomes 
the follower aims to induce via strategic reporting:
While 
SSE 
are defined by strategy profiles, multiple strategy profiles
can lead to the same 
distribution over leaf nodes.
It is the latter that really influences players' expected utilities, and thus is what players really care. 
To realize a desired distribution,
the easiest-to-analyze way for the leader is to commit to minimax strategies on the subgames she does not want the follower to choose. However, there can be other commitments that is hard to describe but also lead to the same outcome. These nonessentials, as we will show, can be ignored by focusing on 
distributions 
instead of strategy profiles. In particular, this greatly simplifies our algorithms for finding optimal (strongly-)inducible distributions.

Second, we use recursive conditions on game trees to avoid non-convexity issues: when we analyze the leader's optimal commitment (when payoff functions are known) and the follower's optimal strategic reporting, consider the sets of all possible 
distributions
that can be achieved by commitments and induced by strategic reporting, respectively. Since the follower always best-responds with pure strategies, these sets are not convex. 
Thus, traditional convex programming approaches, e.g., linear programming, are not applicable in our problem. Instead, based on the sequential structure and the subgame perfect property of best responses, specially owned by extensive-form games, we are able to characterize the aforementioned sets and design algorithms with recursion. 

Third, we successfully identify a subclass of ``succinct'' distributions on leaf nodes, whose support size is at most 2: 
when we want to design algorithms for optimal deception or to find a sufficient and necessary condition whether a distribution can be uniquely induced, this subclass of ``succinct'' distributions are easier to analyze, and maintain the same supremum properties with the universal set. That is, for each distribution that can be (uniquely) induced, one can always find a member of this subclass, that is also (uniquely) induced and yields no less utilities for both players. With this discovery, we are able to reduce an optimization problem among arbitrary distributions to that among a 
clearer-to-analyze subclass of distributions, design algorithms for behavioral commitment settings and characterize games that achieve
``Utility Supremum Equivalence'' property under strong inducibility.

\subsection*{Other Related Work}

There is an increasing interest in the learning and manipulation of private information.
For example, to commit to an optimal auction that maximizes revenue, the seller needs to know the buyers' value distributions and thus how buyers will best-respond to the mechanism~\citep{myerson1981optimal}.  
Since the distributions are actually the buyers' private information, they  may benefit from
pretending a different one by manipulating the bid data \citep{pingzhong2018the, deng2020private, chen2022budget}. In this paper, the private information is the follower's payoff function. 
Our work also relates to the adversarial learning in the machine learning community~\citep{lowd2005adversarial, barreno2010security}, which considers the problem that the training data may be manipulated maliciously, 
leading to inaccurate prediction. 
Thorough studies on extensive-form games with perfect information
would help design more algorithms 
robust to such data manipulation behavior. 

Previous work on the follower's strategic reporting
on SSE 
mainly focuses on normal-form games and security games. As for the normal-form games, \citet{gan2019imitative} studied the case where the follower chooses from a finite set of payoff functions to disguise.
\citet{birmpas2021optimally} considered the 
case where the follower can 
pretend any payoff function. The characterization condition in their work is also maximin value, which coincides with our results on pure commitment and inducibility setting. However, maximin value is not enough in other settings considered in our paper, which attributes to the sequential structures of extensive-form games. 

Although any extensive-form game can be transferred into a normal-form game, results in \citet{birmpas2021optimally} cannot directly be applied to ours. 
First, such transformation has an exponential blow-up in size. 
Second, compared to normal-form games, the sets of all game outcomes that can be induced are not even convex, 
thus linear programming used in \citet{birmpas2021optimally} fails in our case, and new techniques are needed. 

Security games are a special case of 
Stackelberg games with compact  representation, where the defender (the leader) allocates security resources to several targets and the attacker (the follower) chooses one to attack. 
\citet{gan2019manipulating} considered the case where the attacker is able to imitate a different type, while \citet{nguyen2019imitative} studied the case
where the attacker can imitate any payoff function. Their results on ways of imitation show similarities of ours, while totally different techniques are used due to different game structures.  \citet{nguyen2019deception} considered the case in which the attacker with dynamic payoffs can manipulate attacks in repeated security games.  \citet{nguyen2020partial} took the different bounded rationality level of the attacker into consideration. Other studies on the asymmetric information in security games considered how the defenders manipulate when they possess more information~\citep{yin2013optimal,xu2015exploring,rabinovich2015information,guo2017comparing, schlenker2018deceiving}. 

We note that the  
Stackelberg 
equilibria on normal-form 
games and security games 
can be remodeled as
the subgame perfect equilibria (SPE) of a two-layer extensive-form games with perfect information. In this game, the leader acts at root node and her action space there is all her mixed strategies, which is infinite. Our results do not necessarily imply theirs, as we focus on games where the players' action space at each node is finite.

\subsection*{Paper Organization}

\Cref{sec: preliminary} introduces the preliminaries needed throughout this paper. \Cref{sec: pure case} presents the results of the case when the leader commits to pure strategies. 
\Cref{sec:behavioral case} proceeds to 
the case when the leader commits to behavioral strategies. 
\Cref{sec: strong inducibility} provides the results on studying outcomes that can be uniquely induced. 
\Cref{sec: conclusion} summarizes. All omitted proofs are in the appendix.

\section{Preliminaries}
\label{sec: preliminary}

We consider two-player extensive-form games with perfect information on the second phase:
\begin{definition}[Extensive-form Game]
A two-player extensive-form game with perfect information is a tuple 
$(T_{root}, N, \mathcal{P},\{U_i\}_{i\in N})$, where
\begin{itemize}
    \item $T_{root} = (H\cup Z; \operatorname{Child})$ is a rooted game tree, in which  
    \begin{itemize}
        \item $root$ is the root node;
        \item $H$ is the set of decision nodes (internal nodes);
        \item $Z$ is the set of terminal nodes (leaf nodes);
        \item $\operatorname{Child}: H \rightarrow 2^{H\cup Z}$ assigns each internal node $v\in H$ the set of children nodes of $v$, i.e., the set of available actions at $v$;
    \end{itemize}
    \item $N=\{L,F\}$ is the set of players. $L$ ($F$) denotes the \textit{Conceptual Leader} (\textit{Follower}, respectively).
    \item $\mathcal{P}: H \rightarrow N$ assigns each internal node $v \in H$ the unique player to make decision at $v$.
    \item $U_i:Z\rightarrow \mathbb{R}$ is the payoff function of player $i\in N$. 
\end{itemize}

\end{definition}
Since the players $N$ and  assignment $\mathcal{P}$ keeps the same throughout our analysis, we simply denote an extensive-form game by $(T_{root},U_L,U_F)$. 
Let $\mathcal{P}^{-1}(i) = \{ v \in H: \mathcal{P}(v) = i\}$ be the set of decision nodes designated to player $i\in\{L,F\}$. 
A pure strategy 
$\pi_i\in \Pi_i\coloneqq\prod_{v\in \mathcal{P}^{-1}(i)}\child{v}$ 
chooses exactly one action at each node for player $i$ to act. 
A behavioral strategy 
$\delta_i\in \Delta_i\coloneqq \prod_{v\in \mathcal{P}^{-1}(i)}\Delta(\child{v})$
chooses a probability distribution over the action set
$\child{v}$
at each node $v \in \mathcal{P}^{-1}(i)$. 
Note that a pure strategy is also a behavioral strategy. 

Define player $i$'s utility under strategy profile $(\delta_L,\delta_F)$ to be
\begin{equation*}
    U_i(\delta_L,\delta_F) \coloneqq \sum_{z\in Z} U_i(z)p_{(\delta_L,\delta_F)}(z)
\end{equation*}
where $p_{(\delta_L,\delta_F)}(z)$ denotes the probability that the game ends in $z$ under 
$(\delta_L,\delta_F)$. 
Then we can define the best response set of the follower to be\footnote{Since there must be a pure strategy best response for any leader's (pure or behavioral) commitment, we assume the follower always chooses a pure strategy best response~\citep{letchford2010computing}. }
\begin{equation*}
    \BR(\delta_L)={\arg\max}_{\pi_F \in \Pi_F} U_F(\delta_L, \pi_F). 
\end{equation*}

The key solution concept of our paper is the strong 
Stackelberg
equilibrium (SSE). We consider two settings in extensive-form games.  
In one setting, the leader is only allowed to commit to pure strategies. Then we proceed to the general setting when the leader commits to behavioral strategies. 

\begin{definition}[Strong Stackelberg Equilibrium]
Given a game $(T_{root},U_L,U_F)$, 

A \textbf{strong Stackelberg equilibrium with pure commitment} is defined as 
\begin{equation*}
        (\pi_L,\pi_F)\in \mathop{\arg\max}_{{\pi_L'}\in \Pi_L,\pi_F^{\prime}\in \BR({\pi_L'})}U_L({\pi_L'},{\pi_F'});
    \end{equation*}

A \textbf{strong Stackelberg equilibrium with behavioral commitment} is defined as 
    \begin{equation*}
    (\delta_L,\pi_F)\in \mathop{\arg\max}_{{\delta_L'}\in \Delta_L,{\pi_F'}\in \BR({\delta_L'})}U_L({\delta_L'},{\pi_F'}).
    \end{equation*}

\end{definition}

We note that the definitions of SSEs imply that the follower breaks ties in favor of the leader, and then their existence in these two settings is guaranteed 
\citep{letchford2010computing}.
An example in \Cref{example: gain more by committing behavioral} shows that, in contrast to SPEs in games with perfect information, 
the leader can gain more by committing to behavioral strategies in SSEs: when the leader only commits to pure strategies, the best utility she can get is $2$, while she can get expected utility of $2.8$ by behavioral commitment.

\begin{figure}[!tbhp]
\captionsetup[subfigure]{}
\subcaptionbox{SSEs with pure commitment of the game. \label{example: gain more by committing behavioral-pure SSE} }[.49\textwidth]{%
\centering
\begin{tikzpicture}[scale=0.95, every node/.style={transform shape}]
\centering
\node[above left] at (3,1.6) {$F$};
\node[above left] at (1.5,0.8) {$L$};
\node[above right] at (4.5,0.8) {$L$};
\node at (0,-0.3) {$U_L$};
\node at (0,-0.7) {$U_F$};
\node at (0.75,-0.3) {$2$};
\node at (0.75,-0.7) {$3$};
\node at (2.25,-0.3) {$4$};
\node at (2.25,-0.7) {$-2$};
\node at (4.5, 0.5) {$1$};
\node at (4.5,0.1) {$1$};
\node at (4, 0.5) {$U_L$};
\node at (4,0.1) {$U_F$};
\filldraw (3,1.6) circle (.08)
          (1.5,0.8) circle (.08)
          (4.5,0.8) circle (.08)
          (0.75,0) circle (.08)
          (2.25,0) circle (.08);
\draw[thick] (3,1.6)--(1.5,0.8) 
         (3,1.6)--(4.5,0.8) 
         (1.5,0.8)--(0.75,0) 
         (1.5,0.8)--(2.25,0);
\draw[->, thick, blue] (2.8,1.58)--(1.6,0.94);
\draw[->, thick, blue] (1.355,0.75)--(0.755,0.11);
\end{tikzpicture}}
\hfill
\subcaptionbox{SSEs with behavioral commitment of the game. \label{example: gain more by committing behavioral-behavioral SSE}}[.49\textwidth]{
\begin{tikzpicture}[scale=0.95, every node/.style={transform shape}]
\centering
\node[above left] at (3,1.6) {$F$};
\node[above left] at (1.5,0.8) {$L$};
\node[above right] at (4.5,0.8) {$L$};
\node at (0,-0.3) {$U_L$};
\node at (0,-0.7) {$U_F$};
\node at (0.75,-0.3) {$2$};
\node at (0.75,-0.7) {$3$};
\node at (2.25,-0.3) {$4$};
\node at (2.25,-0.7) {$-2$};
\node at (4.5, 0.5) {$1$};
\node at (4.5,0.1) {$1$};
\node at (4, 0.5) {$U_L$};
\node at (4,0.1) {$U_F$};
\filldraw (3,1.6) circle (.08)
          (1.5,0.8) circle (.08)
          (4.5,0.8) circle (.08)
          (0.75,0) circle (.08)
          (2.25,0) circle (.08);
\draw[thick] (3,1.6)--(1.5,0.8) 
         (3,1.6)--(4.5,0.8) 
         (1.5,0.8)--(0.75,0) 
         (1.5,0.8)--(2.25,0);
\draw[->, thick, blue] (2.8,1.58)--(1.6,0.94);
\draw[->, densely dashed, blue] (1.355,0.75)--(0.755,0.11);
\draw[->, densely dashed, blue] (1.645,0.75)--(2.245,0.11);
\node[left,blue] at (1.055, 0.43) {$\frac{3}{5}$};
\node[right,blue] at (1.945, 0.43) {$\frac{2}{5}$};
\end{tikzpicture}}
\caption{An example where the leader gains more utility by committing to behavioral strategies than by only committing to pure strategies. Blue solid arrows show that the players choose pure actions at some nodes. Blue dotted arrows show that the players choose distributions. Fig.~\ref{example: gain more by committing behavioral-pure SSE} shows the SSE with pure commitment of the example game, i.e., $L$ commits to going left and $F$ best-responds with going left.
The leader gets $2$ by optimal pure commitment. Fig.~\ref{example: gain more by committing behavioral-behavioral SSE} shows the SSE with behavioral commitment, i.e., the leader commits to going left with probability $\frac{3}{5}$ and going right with probability $\frac{2}{5}$.
The follower still chooses to go left, because his expected utilities by going both ways are the same, and he breaks ties in favor of the leader. The leader can get $2.8$ in the SSE with behavioral commitment, which is $0.8$ more than the SSE with pure commitment. }
\label{example: gain more by committing behavioral}
\end{figure}
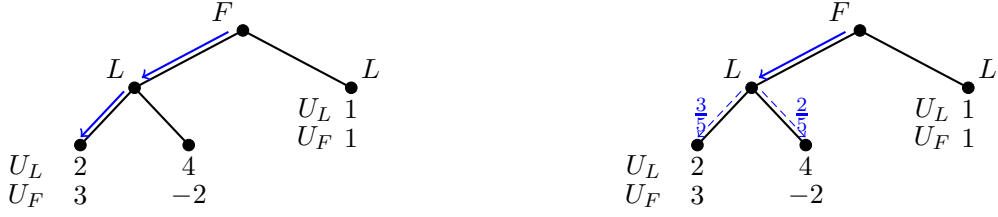

\paragraph{Equilibrium manipulation via strategic reporting}
While any strategy profile $(\delta_L,\pi_F)$ realizes a unique
distribution 
over leaf nodes 
$p_{(\delta_L,\pi_F)} \in \Delta(Z)$, 
a probability distribution may correspond to zero or multiple strategy profiles, as the follower always best-responds with pure strategies. We call a distribution \textbf{realizable} if it corresponds to at least one strategy profile and give it a characterization in \Cref{sec: properties for behavioral case}. 
We will simply call $(\delta_L,\pi_F)$ a strategy profile of distribution $p\in\Delta(Z)$ if $p_{(\delta_L,\pi_F)} = p$. 

Since players' expected utilities are actually determined by these 
distributions, 
we will study all distributions that the follower can report a payoff function and induce an SSE leading to it. We use ``inducibility'', originally proposed in \citet{birmpas2021optimally}, to describe such distributions.

\begin{definition}[Inducibility]
\label{def: induce}

Given game $(T_{root},U_L,U_F)$,  distribution on leaf nodes $p\in\Delta(Z)$ is \textbf{induced by} payoff function $\tilde{U}_F$, 
if a strategy profile $(\delta_L,\pi_F)$ of $p$
is an SSE of game $(T_{root},U_L,\tilde{U}_F)$. We say $p$ is \textbf{inducible}, if such a $\tilde{U}_F$ exists.  

Specifically, $p$ is inducible \textbf{with pure commitment} if $(\delta_L,\pi_F)$ is an SSE with pure commitment; $p$ is inducible \textbf{with behavioral commitment} if $(\delta_L,\pi_F)$ is an SSE with behavioral commitment. 
\end{definition}

\paragraph{Function restrictions on subtrees}
Finally, due to the tree structure of extensive-form games, functions defined on subtrees are often involved in the incoming analysis.  
We denote the subgame tree with root node $v \in H\cup Z$ as $T_v$, the nodes of $T_v$ as $(H\cup Z)_v$, while the internal nodes and terminal nodes as $H_v$ and $Z_v$, respectively. 

Let $f|_{v}$ be a function defined on a subset $S \subseteq (H\cup Z)_v$ for node $v \in H$. For  $w\in\child{v}$, define its restriction on subtree $T_w$ of $T_v$, $f|_{w}$, to be:
\begin{equation*}
    \begin{aligned}
     \operatorname{Dom}  f|_{w}&= S \cap \left(H\cup Z\right)_w\\
     f|_{w}(x)  &= f|_{v}(x),~\forall x \in S \cap \left(H\cup Z\right)_w
    \end{aligned}
\end{equation*}

When $f=\BR$, to simplify notations, noticing that any strategy profile defined on a subtree can be 
regarded as the restriction of one defined on the whole game tree,
we will only use $\pi_F \in \BR(\delta_L)|_{v}$ to denote that $\pi_F|_{v}$ is a best response 
of $\delta_L|_{v}$ at subtree $T_v$. We call such 
$(\delta_L,\pi_F)|_{v}$ \textbf{feasible} in game $(T_{v}, U_L|_{v},U_F|_{v})$. 

We will omit the $U_i$s and simply say a distribution (strategy profile) is inducible (feasible, respectively) at $T_v$ when there is no confusion.

\section{Optimally Reporting under Pure Commitment}
\label{sec: pure case}

This section presents results of the follower's strategic reporting when the leader only commits to pure strategies. 
When $\delta_L$ is a pure strategy $\pi_L$, $(\pi_L,\pi_F)$ leads to a distribution that put all probability on one leaf node and is  equivalent to the one $z\in Z$ that $p_{(\pi_L,\pi_F)}(z) = 1$. Thus we will directly talk about the inducibility of leaf nodes.
We first give some important properties for pure commitment with respect to the maximin values in \Cref{sec: properties for pure case}.  We characterize all the inducible leaf nodes with pure commitment in 
\Cref{sec:cheat result in pure case}. To put it simply, a leaf node is inducible if and only if the leader's utility is no less than her maximin value.
We then 
give a linear-time algorithm
to find an optimal inducible leaf node and construct a 
corresponding follower's
payoff function in 
\Cref{sec: algorithms under pure}.

\subsection{Relating the Maximin Value with Pure Commitments}
\label{sec: properties for pure case}

We first define the maximin value of an extensive-form game as follows.
\begin{definition}[$M_i$]
\label{def: max-min value}
Let $v\in H\cup Z$. We define $M_i(v)$ for $i \in N$ recursively as follows:
\begin{eqnarray}
M_i(v)=\left\{ \begin{array}{ll}
U_i(v),& {v\in Z};\\
\max_{w\in \mathrm{Child}(v)} M_i(w),& v\in \mathcal{P}^{-1}(i);\\
\min_{w\in \mathrm{Child}(v)} M_i(w),& v\in \mathcal{P}^{-1}(-i).
\end{array}
\right.
\end{eqnarray}
\end{definition}

We relate $M_i(v)$ to values players can get by leader's optimal commitment. Specifically,  \Cref{lemma: min utility in pure commit} shows that $M_i(root)$ is the least value players can guarantee themselves on SSEs.

\begin{restatable}{proposition}{lemmaminutilityinpurecommit}\label{lemma: min utility in pure commit}
Given a leaf node $z\in Z$, if there exists an SSE with pure commitment of game $(T_{root}, U_L,U_F)$ leading to $z$, then $U_L(z)\geq M_L(root)$ and $U_F(z)\geq M_F(root)$.
\end{restatable}

The following \Cref{lemma: zero-sum in pure commit} shows that the value of an SSE in a zero-sum extensive-form game, i.e., the optimal value the leader can get by committing, is equal to $M_L(root)$. 

\begin{restatable}{proposition}{lemmazerosuminpurecommitment}\label{lemma: zero-sum in pure commit}
Given a zero-sum extensive-form game $(T_{root},U_L,U_F=-U_L)$, we have \begin{equation*}
    \mathop{\max}_{\pi_L\in \Pi_L,\pi_F\in \BR(\pi_L)}U_L(\pi_L,\pi_F)=M_L(root).
\end{equation*}
\end{restatable}

\Cref{lemma: pure commit} characterizes leaf nodes that the leader can lead the game to via pure commitment. 

\begin{restatable}{proposition}{lemmapurecommit}
\label{lemma: pure commit}
For leaf node $z\in Z$ in game $(T_{root},U_L,U_F)$,  there exists a feasible strategy profile $(\pi_L,\pi_F\in \BR(\pi_L))$ leading to $z$
if and only if $U_F(z)\geq M_F(v)$ for every $v\in \mathcal{P}^{-1}(F)$ that is on the path from $root$ to $z$.
\end{restatable}

\subsection{Leader's Maximin Value Yields all Inducible Leaf Nodes}
\label{sec:cheat result in pure case}

In this section, we first characterize all inducible leaf nodes with pure commitment.  
\begin{restatable}{theorem}{thmcheatinpurecommit}
\label{thm: cheat in pure commit}
Leaf node $z\in Z$ is inducible with pure commitment if and only if $U_L(z)\geq M_L(root)$.
\end{restatable}

Comparing our results with the characterization of all feasible pure strategy profiles in \Cref{lemma: pure commit}, one can find out that with the ability of exploiting his private utility information, the follower has the actual advantage when the leader only commits to pure strategies. Moreover, he gains more power through strategically reporting than committing: To induce a leaf node, the follower only needs to check the leader's utility at $root$, while to check if it can be achieved by commitment, the leader has to consider the follower's maximin values in every subgame with positive probabilities to be played. While this comparison of power is not 
that 
obvious in the behavioral commitment case, as we will show in \Cref{sec: higher utility than pure}, the follower can generally gain more  under behavioral commitment case than under pure case, via misreporting optimally. 

\begin{proof}[Proof of \Cref{thm: cheat in pure commit}]
\textbf{The necessity} is immediately from \Cref{lemma: min utility in pure commit}. A leaf node $z$ is inducible means that for some 
payoff function, $\tilde{U}_F$, an SSE with pure commitment of game $(T_{root}, U_L,\tilde{U}_F)$ leads to $z$. Thus $U_L(z)\geq M_L(root)$. 
We prove \textbf{the sufficiency} by structural induction over the game tree. 

\textbf{Inductive Base:} When $root \in Z$, the only leaf node in $Z$ is $root$, then $U_L(root)=M_L(root)$. Since the strategy set of $T_{root}$ is empty, $root$ is inducible with pure commitment at $T_{root}$.  

\textbf{Inductive Step:} When $root \in H$, suppose that leaf node $z\in Z$ satisfies $U_L(z)\geq M_L(root)$. 

\textbf{When $\mathcal{P}(root)=L$}, $ U_L(z)\geq M_L(root)=\max_{v^\prime\in \child{root}} M_L(v^\prime)$. Let $v\in\child{root}$ be $z$'s ancestor, then $U_L(z)\geq M_L(v)$. So $z$ is inducible at $T_v$
by the inductive hypothesis.  
Let $\tilde{U}_F|_{v}$ be a payoff function defined on $T_v$ that an SSE of subgame $(T_v,U_L|_{v},\tilde{U}_F|_{v})$ leads to $z$, and extend it to $T_{root}$ as follows:
\begin{eqnarray*}
\tilde{U}_F(z^\prime)=\left\{ \begin{array}{ll}
\tilde{U}_F|_{v}(z^\prime) & {z^\prime\in Z_v}\\
-U_L(z^\prime) &  z^\prime\in Z\setminus Z_v.
\end{array}
\right.
\end{eqnarray*}

We prove that 
there is 
an SSE with pure commitment of game $(T_{root},U_L,\tilde{U}_F)$ leading to $z$. Consider any pure commitment $\pi_L\in\Pi_L$ of the leader. If $\pi_L(root)=v$, since an SSE of 
$(T_v,U_L|_{v},\tilde{U}_F|_{v})$ leads to $z$, the leader's utility by commitment is upper bounded by $U_L(z)$. 
The leader can simply choose the same commitment in the SSE of 
$(T_{v},U_L|_{v},\tilde{U}_F|_{v})$ that leads to $z$ and get $U_L(z)$. If $\pi_L(root)=v'\neq v$, then by \Cref{lemma: zero-sum in pure commit},
the leader's utility is upper bounded by 
\begin{equation*}
    \max_{\pi_L\in \Pi_L, \pi_F\in \BR(\pi_L)|_{{v'}}}U_L(\pi_L,\pi_F)|_{{v'}}=M_L(v')\leq U_L(z). 
\end{equation*}
So there is an SSE with pure commitment in 
$(T_{root},U_L,\tilde{U}_F)$ leading to $z$.

\textbf{When $\mathcal{P}(root)=F$}, let $v_1\in\child{root}$ be the ancestor of $z$ and $v_2\in\arg\min_{v'\in\child{root}}M_L(v')$. Note that $v_1$ may equal  $v_2$ and we first consider the case $v_1\neq v_2$. Construct 
$\tilde{U}_F$ as follows:
\begin{eqnarray}
\tilde{U}_F(z^\prime)=\left\{ \begin{array}{ll}
-M_L(v_2)+1 & {z^\prime=z};\\
-M_L(v_2)-1 & z^\prime\not\in Z_{v_2}\wedge z^\prime\neq z;\\
-U_L(z^\prime) & z^\prime\in Z_{v_2}.
\end{array}
\right.
\end{eqnarray}

For any pure commitment $\pi_L\in\Pi_L$ of the leader, 
if the follower cannot achieve $z$ under $\pi_L$, he will choose $v_2$ at $root$ as his best response, obtaining a utility of at least  $-M_L(v_2)$. Then the leader's utility is upper bounded by $M_L(v_2)\leq U_L(z)$ by \Cref{lemma: zero-sum in pure commit}. If the follower can achieve $z$ under 
$\pi_L$, his best response is to achieve $z$ and get $-M_L(v_2)+1$. The leader then gets $U_L(z)$, which is the best she can get. Thus there is an SSE with pure commitment of 
$(T_{root},U_L,\tilde{U}_F)$ leading to $z$. 

If $v_1=v_2$, $U_L(z)\geq M_L(v_1)$ holds, then $z$ is inducible 
at $T_{v_1}$ 
by the inductive hypothesis. Let $\tilde{U}_F|_{{v_1}}$ be
a 
payoff function that induces it at $T_{v_1}$. 

Define $\tilde{m}_F(v_1) = \min_{z^\prime\in Z_{v_1}}\tilde{U}_F|_{{v_1}}(z^\prime)$, 
and extend 
$\tilde{U}_F|_{{v_1}}$ to $T_{root}$ as follows:
\begin{eqnarray}
\tilde{U}_F(z^\prime)=\left\{ \begin{array}{ll}
\tilde{U}_F|_{{v_1}}(z') & {z'\in Z_{v_1}};\\
\tilde{m}_F(v_1)-1 & z'\in Z \setminus Z_{v_1}.
\end{array}
\right.
\end{eqnarray}
For whatever strategy the leader commits to, the follower will always choose $v_1$ at $root$, and get at least $\tilde{m}_F(v_1) > \tilde{m}_F(v_1)-1$.
Then what matters is the strategies' restrictions on $T_{v_1}$. Since an SSE in 
subgame 
$(T_{v_1},U_L|_{{v_1}},\tilde{U}_F|_{{v_1}})$ leads to $z$, there is also an SSE in 
game 
$(T_{root},U_L,\tilde{U}_F)$ leading to $z$. 

This finishes the proof.
\end{proof}

We note that our proof of \Cref{thm: cheat in pure commit} is constructive. With a slight modification, we can actually construct a
payoff function in time $O(|H|+|Z|)$ for any inducible leaf node.

\begin{corollary}\label{coro:construct U_F in pure}
For any inducible leaf node $z\in Z$, we can construct a
payoff function $\tilde{U}_F$ that induces $z$ in time $O(|H|+|Z|)$.
\end{corollary}

\subsection{Algorithm for Optimally Reporting}
\label{sec: algorithms under pure}

With the 
characterization of all inducible leaf nodes (\Cref{thm: cheat in pure commit}), we can compute an optimal inducible leaf node and its 
follower's payoff function in linear time\footnote{
In this paper, we will use random access machine model to determine the algorithms' time complexities. }: compute $M_L(root)$ and traverse all leaf nodes that satisfy conditions in \Cref{thm: cheat in pure commit}. 
\begin{restatable}{theorem}{thmalgorithmforpurecase}
\label{thm: algorithm for pure case}
Computing an inducible leaf node $z^*$ and a follower's payoff function $\tilde{U}_F^*$ such that 
\begin{itemize}
\item a 
strategy profile $(\pi_L,\pi_F)$ of $z^*$ is an SSE with pure commitment of game $(T_{root},U_L,\tilde{U}_F^*)$;
\item $U_F(z^*)\geq U_F(z)$ for all inducible leaf node $z$;
\end{itemize}
can be solved in time $\Theta(|H|+|Z|)$.
\end{restatable}

\section{Optimally Reporting under Behavioral Commitment}
\label{sec:behavioral case}

This section presents results of the follower's strategic reporting when the leader commits to behavioral strategies. 
Due to the game tree structure of extensive-form games, the choice of strategies at a node is closely related to
players' utilities
in the subgames. Thus maximin value on $root$ is not enough for characterizing all inducible distributions. We show the characterization in \Cref{thm: cheat in behavioral commit}. Its constructive proof provides an efficient way to
find a
follower's payoff function for an inducible distribution.

We note that the proof of \Cref{thm: cheat in pure commit} cannot be directly applied to the behavioral case. 
When the leader acts at the root node, different from pure commitment case, $U_L(\delta_L,\pi_F)$ is a convex combination of $U_L(\delta_L,\pi_F)|_{v}$s for $v\in\child{root}$. There is no reason for the leader to put positive probabilities on those nodes that lead to less utilities. Thus the conditions no longer hold for the case $\mathcal{P}(root)=L$, 
and neither does the inductive analysis when $\mathcal{P}(root)=F$.

We use two examples 
in Figure \ref{example: why the pure result can not imply} to further illustrate the necessities of additional conditions in \Cref{thm: cheat in behavioral commit}. In each example, a 
strategy profile of a distribution is presented with red arrows. 
Both distributions 
satisfy conditions in \Cref{thm: cheat in pure commit} but are not inducible in the behavioral case.

\begin{figure}[!htbp]
\captionsetup[subfigure]{}
\subcaptionbox{Example 1 for the necessities of Condition (1)(a). \label{example: necessities of behavioral condition-1}}[.49\textwidth]{%
\centering
\begin{tikzpicture}[scale=0.95, every node/.style={transform shape}]
\centering
\node[above left] at (3,1.6) {$L$};
\node[above left] at (1.5,0.8) {$F_1$};
\node[above right] at (4.5,0.8) {$F_2$};
\node[above left] at (0.75,0) {$z_1$};
\node[above right] at (2.25,0) {$z_2$};
\node[above left] at (3.75,0) {$z_3$};
\node[above right] at (5.25,0) {$z_4$};
\node at (0,-0.3) {$U_L$};
\node at (0.75,-0.3) {$8$};
\node at (2.25,-0.3) {$4$};
\node at (3.75,-0.3) {$2$};
\node at (5.25,-0.3) {$1$};
\filldraw (3,1.6) circle (.08)
          (1.5,0.8) circle (.08)
          (4.5,0.8) circle (.08)
          (0.75,0) circle (.08)
          (2.25,0) circle (.08)
          (3.75,0) circle (.08)
          (5.25,0) circle (.08);
\draw[thick] (3,1.6)--(1.5,0.8)
      (3,1.6)--(4.5,0.8)
      (1.5,0.8)--(0.75,0)
      (1.5,0.8)--(2.25,0)
      (4.5,0.8)--(3.75,0)
      (4.5,0.8)--(5.25,0);
\draw[->, thick, densely dashed, red] (3.2,1.58)--(4.4,0.94);
\draw[->, thick, densely dashed, red] (2.8,1.58)--(1.6,0.94);
\draw[->, thick, red] (4.355,0.75)--(3.755,0.11);
\draw[->, thick, red] (1.355,0.75)--(0.755,0.11);
\node[left,red] at (4.3, 1.33) {$\frac{1}{2}$};
\node[right, red] at (1.75,1.33) {$\frac{1}{2}$};
\end{tikzpicture}}
\hfill
\subcaptionbox{Example 2 for the necessities of condition (2)(b). \label{example: necessities of behavioral condition-2}}[.49\textwidth]{
\begin{tikzpicture}[scale=0.95, every node/.style={transform shape}]
\centering
\node[above left] at (2.75,2.4) {$F_1$};
\node[above left] at (1.35,1.6) {$L_1$};
\node[above right] at (4.35,1.6) {$L_2$};
\node[above left] at (3.20,0.8) {$F_2$};
\node[above right] at (5.5,0.8) {$F_3$};
\node[above left] at (0.84,0.8) {$z_1$};
\node[above right] at (1.77,0.8) {$z_2$};
\node[above left] at (2.68,0) {$z_3$};
\node[above right] at (3.71,0) {$z_4$};
\node[above left] at (4.98,0) {$z_5$};
\node[above right] at (6.01,0) {$z_6$};
\node at (0.35,0.5) {$U_L$};
\node at (0.8,0.5) {$8$};
\node at (1.9,0.5) {$4$};
\node at (2.65,-0.3) {$8$};
\node at (3.75,-0.3) {$4$};
\node at (4.95,-0.3) {$2$};
\node at (6.05,-0.3) {$1$};
\filldraw (2.75,2.4) circle (.08)
          (1.35,1.6) circle (.08)
          (4.35,1.6) circle (.08)
          (0.8,0.8) circle (.08)
          (1.9,0.8) circle (.08)
          (3.2,0.8) circle (.08)
          (5.5,0.8) circle (.08)
          (2.65,0) circle (.08)
          (3.75,0) circle (.08)
          (4.95,0) circle (.08)
          (6.05,0) circle (.08);
\draw[thick] (2.75,2.4)--(1.35,1.6)
      (2.75,2.4)--(4.35,1.6)
      (1.35,1.6)--(0.8,0.8)
      (1.35,1.6)--(1.9,0.8)
      (4.35,1.6)--(3.2,0.8)
      (4.35,1.6)--(5.5,0.8)
      (3.20,0.8)--(2.65,0)
      (3.20,0.8)--(3.75,0)
      (5.5,0.8)--(4.95,0)
      (5.5,0.8)--(6.05,0);
\draw[->, thick, red] (3.01,2.375)--(4.26,1.75);
\draw[->, thick, densely dashed, red] (4.1, 1.52)--(3.28,0.95);
\draw[->, thick, densely dashed, red] (4.58, 1.52)--(5.42,0.95);
\draw[->, thick, red] (3.04,0.7)--(2.66, 0.15);
\draw[->, thick, red] (5.34,0.7)--(4.96, 0.15);
\node[left,red] at (3.69, 1.31) {$\frac{1}{2}$};
\node[right, red] at (5.0,1.31) {$\frac{1}{2}$};
\end{tikzpicture}}

\caption{Two examples illustrating the necessities of additional conditions in \Cref{thm: cheat in behavioral commit}. Red arrows in each subfigure represent a 
strategy profile of a distribution $p$ that satisfies the conditions in \Cref{thm: cheat in pure commit} but is not inducible with behavioral commitment. That is, $U_L(p)\geq M_L(root)$.  
As for the distribution and the game shown in Fig.~\ref{example: necessities of behavioral condition-1}, $p(z_1)=p(z_3)=\frac{1}{2}$. When the leader acts at $root$, going left with probability $1$ is always a dominant strategy. There is no need for her to commit to strategies which are mixed at $root$, although she can get utilities greater than $M_L(root)=4$.
In Fig.~\ref{example: necessities of behavioral condition-2}, $p(z_3)=p(z_5)=\frac{1}{2}$. 
For any payoff function $\tilde{U}_F$ that makes the red strategy profile feasible, by \Cref{lemma: feasible}, going to $z_3$ at $F_2$ is the follower's best response at subgame $T_{F_2}$. Then by going left at $L_1$ and $L_2$, the leader can get $8>5$, 
and the red strategy profile will never be an SSE no matter the follower's reported payoff function. }
\label{example: why the pure result can not imply}
\end{figure}
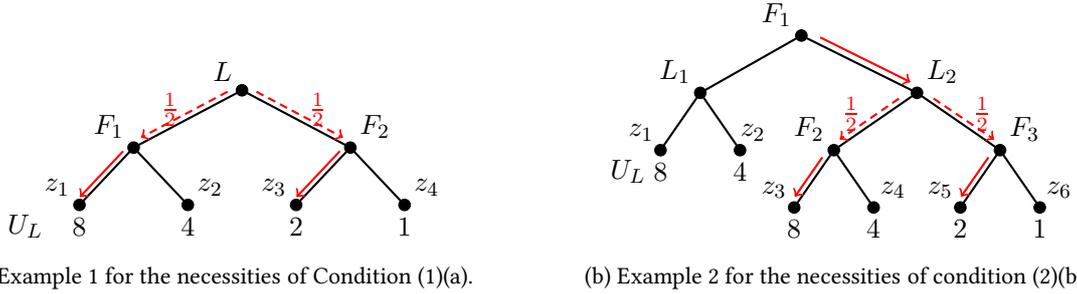

We 
provide
a polynomial-time algorithm to find an optimal inducible distribution with behavioral commitment in \Cref{sec: algorithms under behavioral}. Actually, while the characterization conditions are complicated, 
\Cref{thm: cheat in behavioral commit} helps us find optimal inducible distributions with succinct structure, which we call ``Y-shape'' as defined in \Cref{def:y-shape inducible}.
Furthermore, we show that the follower can gain no less utility by strategic reporting with behavioral commitment than with pure commitment in \Cref{sec: higher utility than pure}.

\subsection{Additional Notations}
Given distribution $p\in\Delta(Z)$
, define $\operatorname{Supp}(p)\coloneqq\{z\in Z: p(z)>0\}$. 
With a slight abuse of notation,
we recursively extent it onto tree $T_{root}$ as follows, 
\begin{eqnarray}
p(v)=\left\{ \begin{array}{ll}
p(v),&{v\in Z};\\
\sum_{w\in \mathrm{Child}(v)} p(w),&v\in H. 
\end{array}
\right.
\end{eqnarray}
For each $v\in H\cup Z$, $p(v)$ denotes the probability of reaching $v$ from $root$ to realize $p$ on $Z$.

For the incoming inductive proofs, we define the restriction of $p$ on subtree $T_{v}$, for $v$ with $p(v)>0$, to be $p|_{v}\in \Delta(Z_v)$ and  $p|_{v}(z)\coloneqq p(z)/p(v), \forall z\in Z_v$. 
When $v=root$, we simply use $p$ instead of $p|_{{root}}$. The notions of inducibility 
and utility on subgame $(T_v,U_L|_{v},U_F|_{v})$ are defined similarly.

For a distribution $p|_{v} \in \Delta(Z_v)$ and $w\in H_v$, we use
\begin{equation*}
\mathrm{Supp}(p|_{v},w)\coloneqq \left\{x\in\child{w}:p|_{v}(x)>0\right\}
\end{equation*}
to denote the set of $w$'s children that are reached with positive probabilities from $v$ to realize $p|_{v}$.

For a strategy profile $(\delta_L,\pi_F)|_{v}$ on subtree $T_v$, we use $p_{(\delta_L,\pi_F)}|_v$ to denote the its 
distribution on $Z_v$. Similarly, for $w \in H_v$,
we use
\begin{equation*}
        \operatorname{Supp}((\delta_L,\pi_F)|_{v},w)\coloneqq \left\{ x\in \child{w}: p_{(\delta_L,\pi_F)}|_{v}(x) > 0\right\}
    \end{equation*}
to denote the set of $w$'s children that are reached with positive probabilities from $v$ under $(\delta_L,\pi_F)|_{v}$. 

\subsection{Relating the Maximin Value with Behavioral Commitments}
\label{sec: properties for behavioral case}

We extend results on properties for pure commitments in \Cref{sec: properties for pure case}  to behavioral commitments.

\begin{proposition}\label{lemma: min value in behavioral commit}
If  strategy profile 
$(\delta_L,\pi_F)$ is an SSE with behavioral commitment of game $(T_{root}, U_L,U_F)$, then $U_L(\delta_L,\pi_F)\geq M_L(root)$ and $U_F(\delta_L,\pi_F)\geq M_F(root)$. 
\end{proposition}

The proof of  \Cref{lemma: min value in behavioral commit} is similar to that of \Cref{lemma: min utility in pure commit}. Just notice that the forms of the follower's strategies minimizing the leader's utility or maximizing his own utility, given the leader's strategy, do not depend on whether the leader's strategy is pure or behavioral: given the utilities players can get in the subgames, choosing the node with the leader's lowest utility (the follower's highest utility, respectively) is always the correct action. 

Using similar proof to the one of \Cref{lemma: zero-sum in pure commit}, we have the following \Cref{lemma: zero-sum in behavioral commit}, which shows that the leader gains no advantage in committing to behavioral strategies in zero-sum games.

\begin{proposition}\label{lemma: zero-sum in behavioral commit}
Given a zero-sum extensive-form game $(T_{root},U_L,U_F=-U_L)$, we have 
\begin{equation*}
    \mathop{\max}_{\delta_L\in \Delta_L,\pi_F\in \BR(\delta_L)}U_L(\delta_L,\pi_F)=M_L(root).
\end{equation*}
\end{proposition}

We also note the following simple fact:

\begin{restatable}{fact}{lemmacharacterizefeasibledistributions}
\label{lemma: characterize feasible distributions}
$p\in \Delta(Z)$ is realizable if and only if $|\operatorname{Supp}(p,v)|=1$, for every $ v \in \mathcal{P}^{-1}(F)$ with $p(v)>0$. 
\end{restatable}

\subsection{Characterizing all Inducible Distributions by Subtree Recursion}
\label{sec: cheat in behavioral case}

The characterization of all inducible distributions under behavioral commitment is more complicated than that under pure commitment. 
For $root \in Z$, the game tree only has one terminal node and the only distribution $p$ that $p(root)=1$ is inducible. For general extensive-form games, we have the following theorem.

\begin{restatable}{theorem}{thmcheatinbehavioralcommit}
\label{thm: cheat in behavioral commit}
A realizable distribution on leaf nodes $p\in \Delta(Z)$ is inducible if and only if 
\begin{enumerate}
    \item if $\mathcal{P}(root)=L$, all of the following conditions are met:
    \begin{enumerate}
        \item $U_L(p)|_{v}$s are the same at $T_v$ for all $v\in \operatorname{Supp}(p,root)
        $;
        \item $U_L(p)\geq M_L(root)$;
        \item $p|_{v}$ is inducible at 
        $T_v$, 
        for all $v \in \operatorname{Supp}(p,root)$. 
    \end{enumerate}
    \item if $\mathcal{P}(root)=F$, 
     let $\operatorname{Supp}(p,root)=\{v\}$, then at least one of following two conditions are met: 
    \begin{enumerate}
        \item $U_L(p)\geq \min_{v^\prime \in \child{root}, v^\prime\ne v} M_L(v^\prime)$;
        \item $p|_{v}$ is inducible at 
        $T_v$.
    \end{enumerate}
\end{enumerate}
\end{restatable}

 We note that condition (1)(a) in \Cref{thm: cheat in behavioral commit} deals with the case that the leader may mix at $root$ among several children nodes that yield same utilities for her. 
 Recall that a strategy profile $(\delta_L,\pi_F)|_{v}$ is feasible 
 with respect to $(T_v,U_F|_{v})$
 if $\pi_F\in \BR(\delta_L)|_{v}$ and we simply say it is feasible at $T_v$ when there is no confusion. 
  Before proving \Cref{thm: cheat in behavioral commit}, we present a necessary lemma on the properties of feasible strategy profiles. 

\begin{restatable}{lemma}{lemmafeasible}
\label{lemma: feasible}
Given strategy profile $(\delta_L,\pi_F)|_{v}$ defined on subgame $(T_v,U_L|_{v},U_F|_{v})$ for $v\in H$:

\noindent 1.  If $(\delta_L,\pi_F)|_{v}$ is feasible at $T_v$, $(\delta_L,\pi_F)|_{w}$ is also feasible at $T_w$, for any $w\in \operatorname{Supp}((\delta_L,\pi_F)|_{v},v)$.

\noindent 2. If $(\delta_L,\pi_F)|_{w}$ is feasible at $T_w$, for any $w\in \operatorname{Supp}((\delta_L,\pi_F)|_{v},v)$, then $(\delta_L,\pi_F)|_{v}$ is feasible at $T_v$.

\noindent 3. If $(\delta_L,\pi_F)|_{v}$ is feasible and  $U_L(\delta_L,\pi_F)|_{v} = \max_{\pi^\prime_F\in\BR(\delta_L)|_{v}}U_L(\delta_L,\pi^\prime_F)|_{v}$, then $U_L(\delta_L,\pi_F)|_{w}=\max_{\pi^\prime_F\in\BR(\delta_L)|_{w}}U_L(\delta_L,\pi^\prime_F)|_{w}$, for any $w\in \operatorname{Supp}((\delta_L,\pi_F)|_{v},v)$.
\end{restatable}

Now we give the detailed proof of \Cref{thm: cheat in behavioral commit}. Though the conditions in \Cref{thm: cheat in behavioral commit} are more complicated, the 
constructions of payoff functions 
yield similar intuitive idea. That is,
\begin{itemize}
    \item show stronger conflicts of interests where the leader can gain more utilities;
    \item use constant-sum subgames and games with constant worst utilities to restrict the feasible strategy profiles for the leader to consider to commit to;
\end{itemize}

\begin{proof}[Proof of \Cref{thm: cheat in behavioral commit}]
We first prove the case \textbf{$\mathcal{P}(root) = L$}. 

 \textbf{The necessity:} Let $p$ be inducible and $(\delta_L,\pi_F)$  the 
 strategy profile of $p$ that is an SSE of game $(T_{root},U_L,\tilde{U}_F)$ for some payoff function $\tilde{U}_F$, then condition (b) holds by \Cref{lemma: min value in behavioral commit}. 

For condition (a): If  $U_L(\delta_L,\pi_F)|_{{v_1}}>U_L(\delta_L,\pi_F)|_{{v_2}}$ for some $v_1$, $v_2 \in \operatorname{Supp}((\delta_L,\pi_F),root)$, then for any follower's payoff function $\tilde{U}^\prime_F$ that makes $(\delta_L,\pi_F)|_{v}$ feasible at $T_v$, $\forall v \in \operatorname{Supp}((\delta_L,\pi_F),root)$, the leader can always increase the probability of choosing $v_1$ and decrease the probability of choosing $v_2$ at $root$ to gain higher utility than $U_L(\delta_L,\pi_F)$, which means $(\delta_L,\pi_F)$ is not an SSE. 

For condition (c): If $p|_{{v_0}}$ is not inducible at $T_{v_0}$, for some $v_0 \in \operatorname{Supp}((\delta_L,\pi_F),root)$, then for any 
$\tilde{U}^\prime_F$ that makes $(\delta_L,\pi_F)|_{{v_0}}$ feasible at $T_{v_0}$, there always exists a feasible strategy profile $({\delta_L'},{\pi_F'})|_{{v_0}}$ such that $U_L({\delta_L'},{\pi_F'})|_{{v_0}} > U_L(\delta_L,\pi_F)|_{{v_0}} = U_L(\delta_L,\pi_F)$, where the second equality holds by condition (a).  
Extend $({\delta_L'},{\pi_F'})|_{{v_0}}$ to $T_{root}$ by letting $\delta_L'(root)=v_0$. 
$({\delta_L'},{\pi_F'})$ is feasible at $T_{root}$ by \Cref{lemma: feasible}, and,
\begin{equation*}
    U_L({\delta_L'},{\pi_F'})=U_L(\delta^\prime_L,\pi^\prime_F)|_{{v_0}} > U_L(\delta_L,\pi_F)|_{{v_0}} = U_L(\delta_L,\pi_F). 
\end{equation*}
 Thus $(\delta_L,\pi_F)$ is not  SSE of game $(T_{root}, U_L,\tilde{U}_F)$ for any $\tilde{U}_F$, leading to a contradiction. 

\textbf{The sufficiency:} Suppose $p$ satisfies the three conditions. Then $p|_{v}$ is inducible at $T_v$, $\forall v\in \operatorname{Supp}(p,root)$: there exists a 
strategy profile $(\delta_L,\pi_F)|_{v}$ of $p|_{v}$ and a payoff function 
$\tilde{U}_F|_{v}$ at $T_v$, such that  $(\delta_L,\pi_F)|_{v}$ is an SSE of subgame $(T_v,U_L|_{v},\tilde{U}_F|_{v})$. 

 Extend $\tilde{U}_F|_{v}$s ($v\in \operatorname{Supp}(p,root)$) to $T_{root}$ as follows:
\begin{eqnarray}
\label{calculate U_F1}
\tilde{U}_F(z)=\left\{ \begin{array}{ll}
\tilde{U}_F|_{v}(z), & {z\in Z_v \wedge v \in \operatorname{Supp}(p,root)};\\
-U_L(z), & z\in Z_v\wedge \left(v\in \child{root}\setminus\operatorname{Supp}(p,root)\right).
\end{array}
\right.
\end{eqnarray}
Extend 
$(\delta_L,\pi_F)|_{v}$s to $T_{root}$ by letting $(\delta_L,\pi_F)(root)=\sum_{v\in \child{root}}p(v)v$, and $(\delta_L,\pi_F)|_{v}$ be any feasible strategy profile in 
$(T_v,U_L|_v,\tilde{U}_F|_v)$ 
for $v\in \child{root}\setminus\operatorname{Supp}(p,root)$. Then $(\delta_L,\pi_F)$ corresponds to $p$ and is feasible in game  $(T_{root},U_L, \tilde{U}_F)$ by \Cref{lemma: feasible}. At $T_v$s for $v\in \child{root}\setminus\operatorname{Supp}(p,root)$, by \Cref{lemma: zero-sum in behavioral commit}, the leader can get at most $M_L(v)\leq M_L(root)\leq U_L(p)=U_L(\delta_L,\pi_F)$  via commitment, which is also the best she can get at $T_v$s where $v\in \operatorname{Supp}(p,root)$. Thus the leader's optimal commitment is to put all non-zero probabilities only on nodes $v\in \operatorname{Supp}(p,root)$ and that $M_L(v) = U_L(\delta_L,\pi_F)$, and get $U_L(\delta_L,\pi_F)$. Since $(\delta_L,\pi_F)$ is feasible at $T_{root}$, $(\delta_L,\pi_F)$ is an SSE with behavioral commitment of game $(T_{root}, U_L,\tilde{U}_F)$, and thus $p$ is inducible.

When $\mathcal{P}(root)=F$, first we prove the necessity part. 

\textbf{The necessity: }Suppose $U_L(p) < \min_{v^\prime \in \child{root}, v^\prime\ne v} M_L(v^\prime)$. If $p|_{v}$ is not inducible at 
$T_v$, 
then for any $p$'s 
strategy profile $(\delta_L,\pi_F)$ and  
function $\tilde{U}_F|_{v}$ that makes $(\delta_L,\pi_F)|_{v}$ feasible at $T_v$, there exists a feasible 
$({\delta_L'},{\pi_F'})|_{v}$ that  
$U_L({\delta_L'},{\pi_F'})|_{v} > U_L(\delta_L,\pi_F)|_{v}$. 

Then for any 
$\tilde{U}_F$ that makes $(\delta_L,\pi_F)$ feasible at $T_{root}$, let $\delta^\prime_L|_{v}$ be the leader's strategy in the aforementioned strategy profile with respect to $\tilde{U}_F|_{v}$, and extend it to $\delta^\prime_L$ at $T_{root}$ by letting it
equal the leader's strategy in an SSE of subgame $(T_{v^\prime},U_L|_{{v^\prime}},\tilde{U}_F|_{{v^\prime}})$ at $T_{v^\prime}$, for $v^\prime \in \child{root}$, $v^\prime \ne v$; and equal ${\delta^\prime_L}|_{v}$ at $T_v$. Then for ${\pi^\prime_F}\in\BR({\delta^\prime_L})$ that maximizes $U_L({\delta_L'},{\pi_F'})$ 
among the follower's strategies in $\BR({\delta_L'})$, if ${\pi_F'}(root) = v$, 
$U_L({\delta_L'},{\pi_F'})\geq U_L(\delta^\prime_L,\pi^\prime_F)|_{v} > U_L(\delta_L,\pi_F)|_{v}=U_L(\delta_L,\pi_F)$. If ${\pi_F'}(root) = v^\prime\ne v$, the leader can get at least $M_L(v^\prime) > U_L(\delta_L,\pi_F)$, contradicting to the assumption that $p$ is inducible. 

\textbf{The sufficiency:} 
\textbf{Case 1: $U_L(p)\geq \min_{v^\prime \in \child{root}, v^\prime\ne v} M_L(v^\prime)$. } Consider a 
strategy profile 
$(\delta_L,\pi_F)$ of $p$, which equals the leader's minimax strategy and the follower's maximin strategy on nodes $h\in H$ that $ p(h)=0$. Let 
\begin{equation*}
    \begin{aligned}
    v_0 &\in {\arg \min}_{v^\prime \in \child{root}, v^\prime\ne v} M_L(v^\prime),\\
    m_0 &= \min_{v^\prime \in \child{root}, v^\prime\ne v} M_L(v^\prime),\\
    M &= \max\{\max_{z\in Z} U_L(z)+1,1\}\\
    A(\delta_L,\pi_F)&\coloneqq \{z \in Z: p_{(\delta_L,\pi_F)}(z)>0\}=\operatorname{Supp}(p).\\
    \end{aligned}
\end{equation*}
 We define payoff function, $\tilde{U}_F$, as follows:
 
\begin{equation}
\label{calculate U_F2}
\tilde{U}_F(z)=\left\{ \begin{array}{ll}
-U_L(z),&{z\in A(\delta_L,\pi_F)};\\
-U_L(z)+m_0-U_L(\delta_L,\pi_F), &z\in Z_{v_0};\\
-2M, &z\not\in Z_{v_0} \wedge z\not\in A(\delta_L,\pi_F).\\
\end{array}
\right.
\end{equation}
Then $(\delta_L,\pi_F)|_{v_0}$ is an SSE of subgame $(T_{v_0},U_L|_{{v_0}},\tilde{U}_F|_{{v_0}})$, and the leader's (follower's) utility is $m_0$ ($-U_L(\delta_L,\pi_F)$, respectively). 

First, $\pi_F\in\BR(\delta_L)$ in game $(T_{root}, U_L,\tilde{U}_F)$: At any node $h\in \mathcal{P}^{-1}(F)$, if the follower chooses another child $v^\prime$ of $h$ instead of $(\delta_L,\pi_F)(h)$, since all the leaf nodes in $T_{v^\prime}$ does not belong to $A(\delta_L,\pi_F)$, the follower will get at most $\max\{-U_L(\delta_L,\pi_F), -2M\} = -U_L(\delta_L,\pi_F)$. 

Now we prove $U_L(\delta_L,\pi_F)$ is the best the leader can get via commitment, and thus $(\delta_L,\pi_F)$ is an SSE of the game $(T_{root}, U_L,\tilde{U}_F)$. $p$ is inducible. 

First for all feasible 
$({\delta_L'},{\pi_F'}\in\BR({\delta_L'}))$s, 
we have 
\begin{equation}
    \tilde{U}_F({\delta_L'},{\pi_F'})\geq \tilde{U}_F(\delta_L,\pi_F) = -U_L(\delta_L,\pi_F),
\end{equation} since otherwise the follower can always get higher utility by choosing $v_0$ at $root$, yielding a leader's utility of $ m_0 \leq U_L(\delta_L,\pi_F)$. 

If $A({\delta_L'},{\pi_F'})\subseteq A(\delta_L,\pi_F)$, then the leader's utility is the negation of the follower's utility, and thus is at most $U_L(\delta_L,\pi_F)$. 

If $A({\delta_L'},{\pi_F'})\ne A(\delta_L,\pi_F)$, suppose with probability $\alpha \in (0,1)$, the game ends in $Z\setminus A(\delta_L,\pi_F)$, then
\begin{equation}
    \tilde{U}_F({\delta_L'},{\pi_F'}) = -2\alpha M -(1-\alpha) U \geq \tilde{U}_F(\delta_L,\pi_F),
\end{equation}
for some $U\in \mathbb{R}$ representing the leader's expected utility if the game ends in $A(\delta_L,\pi_F)$, then \begin{equation}
    U_L({\delta_L'},{\pi_F'}) < \alpha M + (1-\alpha)U \leq 2\alpha M + (1-\alpha) U \leq - \tilde{U}_F(\delta_L,\pi_F) = U_L(\delta_L,\pi_F).
\end{equation}

\textbf{Case 2:} $p|_{v}$ is inducible at $T_v$. 
Then there exists a 
strategy profile
$(\delta_L,\pi_F)|_{v}$ of $p|_{v}$ and a 
$\tilde{U}_F|_{v}$ at $T_v$, such that $(\delta_L,\pi_F)|_{v}$ is an SSE of subgame $(T_{v},U_L|_{v},\tilde{U}_F|_{v})$. 
Define $\tilde{m}_F = \min_{z\in Z_v}\tilde{U}_F(z)$. Construct payoff function, $\tilde{U}_F$, as follows:
\begin{equation}
\label{calculate U_F3}
\tilde{U}_F(z)=\left\{ \begin{array}{ll}
\tilde{U}_F|_{v}(z) & {z\in Z_v};\\
\tilde{m}_F-1 & z\not\in Z_v.
\end{array}
\right.
\end{equation}
 Then whatever strategy the leader commits to, the follower will always choose $v$ at $root$. Extend $(\delta_L,\pi_F)|_{v}$ to $T_{root}$ by letting $(\delta_L,\pi_F)(root)=v$. Since $U_L(\delta_L,\pi_F)$ is the highest value the leader can get by commitment, $(\delta_L,\pi_F)$ is an SSE of game $(T_{root}, U_L,\tilde{U}_F)$. 

\end{proof}

The proof of \Cref{thm: cheat in behavioral commit} actually gives a polynomial-time algorithm to construct a
payoff function for each inducible distribution, as shown in \Cref{coro: construct U_F in behavioral}. 
\begin{restatable}{corollary}{coroconstructuf}
\label{coro: construct U_F in behavioral}
For any inducible 
distribution over leaf nodes $p\in \Delta(Z)$, we can construct a 
payoff function $\tilde{U}_F$ 
that induces $p$
in $O(|H|\cdot|Z|)$ time. 
\end{restatable}

\subsection{Optimal Inducible Distributions with Succinct Structures}
\label{sec: simple inducible strategy-behavioral}
Although inducible distributions are much more complicated, in this section, 
we show that an optimal inducible distribution can have succinct structure, as we call ``Y-shape''.

\begin{definition}[Y-shape]
\label{def:y-shape inducible}
A  distribution $p\in\Delta(Z)$ is called ``Y-shape'' if $|\operatorname{Supp}(p)|\leq 2$.
\end{definition}

Intuitively, a Y-shape realizable distribution corresponds to strategy profiles, under which the nodes with positive probabilities of being reached from $root$ node are in a ``Y'' shape. This means that the leader will mix two actions at at most one decision node, so the players' utilities is a mixture of utilities at at most two leaf nodes. A demonstration of 
Y-shape distributions is in Figure \ref{example: Y-shape}.

\begin{figure}[!htbp]
\captionsetup[subfigure]{}
\subcaptionbox{A non-Y-shape distribution and one of its 
strategy profiles }[.49\textwidth]{%
\centering
\begin{tikzpicture}[scale=0.95, every node/.style={transform shape}]
\centering

\node[above left] at (2.75,2.4) {$F_1$};
\node[above left] at (1.35,1.6) {$L_1$};
\node[above right] at (4.35,1.6) {$L_2$};
\node[above left] at (3.20,0.8) {$L_3$};
\node[above right] at (5.5,0.8) {$F_2$};
\node[above left] at (0.84,0.8) {$z_1$};
\node[above right] at (1.77,0.8) {$z_2$};
\node[above left] at (2.68,0) {$z_3$};
\node[above right] at (3.71,0) {$z_4$};
\node[above left] at (4.98,0) {$z_5$};
\node[above right] at (6.01,0) {$z_6$};
\node at (0.35,0.4) {$p$};
\node at (0.8,0.4) {$0$};
\node at (1.9,0.4) {$0$};
\node at (2.65,-0.4) {$\frac{1}{6}$};
\node at (3.75,-0.4) {$\frac{1}{3}$};
\node at (4.95,-0.4) {$\frac{1}{2}$};
\node at (6.05,-0.4) {$0$};
\filldraw (2.75,2.4) circle (.08)
          (1.35,1.6) circle (.08)
          (4.35,1.6) circle (.08)
          (0.8,0.8) circle (.08)
          (1.9,0.8) circle (.08)
          (3.2,0.8) circle (.08)
          (5.5,0.8) circle (.08)
          (2.65,0) circle (.08)
          (3.75,0) circle (.08)
          (4.95,0) circle (.08)
          (6.05,0) circle (.08);
\draw[thick] (2.75,2.4)--(1.35,1.6)
      (2.75,2.4)--(4.35,1.6)
      (1.35,1.6)--(0.8,0.8)
      (1.35,1.6)--(1.9,0.8)
      (4.35,1.6)--(3.2,0.8)
      (4.35,1.6)--(5.5,0.8)
      (3.20,0.8)--(2.65,0)
      (3.20,0.8)--(3.75,0)
      (5.5,0.8)--(4.95,0)
      (5.5,0.8)--(6.05,0);
\draw[->, thick, red] (3.01,2.375)--(4.26,1.75);
\draw[->, thick, densely dashed, red] (4.1, 1.52)--(3.28,0.95);
\draw[->, thick, densely dashed, red] (4.58, 1.52)--(5.42,0.95);
\draw[->, thick, densely dashed, red] (3.04,0.7)--(2.66, 0.15);
\draw[->, thick, densely dashed, red] (3.36,0.7)--(3.74, 0.15);
\draw[->, thick, red] (5.34,0.7)--(4.96, 0.15);
\node[left,red] at (3.69, 1.31) {$\frac{1}{2}$};
\node[right, red] at (5.0,1.31) {$\frac{1}{2}$};
\node[left, red] at (2.9,0.44) {$\frac{1}{3}$};
\node[right, red] at (3.5, 0.44) {$\frac{2}{3}$};
\end{tikzpicture}}
\hfill
\subcaptionbox{A Y-shape distribution and one of its 
strategy profiles}[.49\textwidth]{
\begin{tikzpicture}[scale=0.95, every node/.style={transform shape}]
\centering

\node[above left] at (2.75,2.4) {$F_1$};
\node[above left] at (1.35,1.6) {$L_1$};
\node[above right] at (4.35,1.6) {$L_2$};
\node[above left] at (3.20,0.8) {$L_3$};
\node[above right] at (5.5,0.8) {$F_2$};
\node[above left] at (0.84,0.8) {$z_1$};
\node[above right] at (1.77,0.8) {$z_2$};
\node[above left] at (2.68,0) {$z_3$};
\node[above right] at (3.71,0) {$z_4$};
\node[above left] at (4.98,0) {$z_5$};
\node[above right] at (6.01,0) {$z_6$};
\node at (0.35,0.4) {$p$};
\node at (0.8,0.4) {$0$};
\node at (1.9,0.4) {$0$};
\node at (2.65,-0.4) {$\frac{1}{2}$};
\node at (3.75,-0.4) {$0$};
\node at (4.95,-0.4) {$\frac{1}{2}$};
\node at (6.05,-0.4) {$0$};
\filldraw (2.75,2.4) circle (.08)
          (1.35,1.6) circle (.08)
          (4.35,1.6) circle (.08)
          (0.8,0.8) circle (.08)
          (1.9,0.8) circle (.08)
          (3.2,0.8) circle (.08)
          (5.5,0.8) circle (.08)
          (2.65,0) circle (.08)
          (3.75,0) circle (.08)
          (4.95,0) circle (.08)
          (6.05,0) circle (.08);
\draw[thick] (2.75,2.4)--(1.35,1.6)
      (2.75,2.4)--(4.35,1.6)
      (1.35,1.6)--(0.8,0.8)
      (1.35,1.6)--(1.9,0.8)
      (4.35,1.6)--(3.2,0.8)
      (4.35,1.6)--(5.5,0.8)
      (3.20,0.8)--(2.65,0)
      (3.20,0.8)--(3.75,0)
      (5.5,0.8)--(4.95,0)
      (5.5,0.8)--(6.05,0);
\draw[->, thick, red] (3.01,2.375)--(4.26,1.75);
\draw[->, thick, densely dashed, red] (4.1, 1.52)--(3.28,0.95);
\draw[->, thick, densely dashed, red] (4.58, 1.52)--(5.42,0.95);
\draw[->, thick, red] (3.04,0.7)--(2.66, 0.15);
\draw[->, thick, red] (5.34,0.7)--(4.96, 0.15);
\node[left,red] at (3.69, 1.31) {$\frac{1}{2}$};
\node[right, red] at (5.0,1.31) {$\frac{1}{2}$};
\end{tikzpicture}}

\caption{Examples to illustrate Y-shape distributions.}
\label{example: Y-shape}
\end{figure}
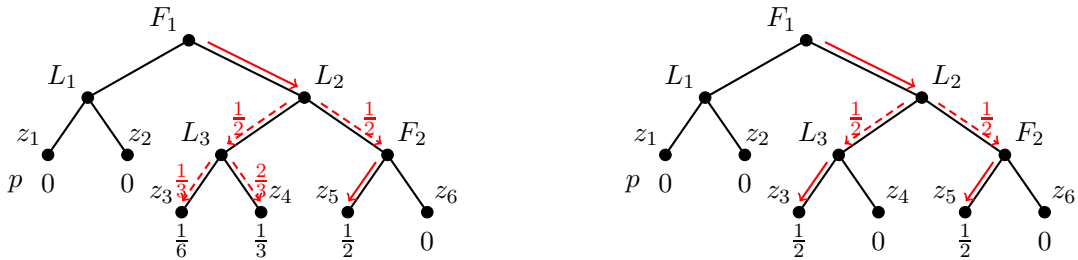

Now we utilitize \Cref{thm: cheat in behavioral commit} to show that for any inducible distribution $p$, there is always a dominant ``Y-shape'' inducible distribution, in the sense that the follower can always get no less utility on it. Thus there exists an optimal ``Y-shape'' inducible distribution. \Cref{lemma: simple inducible sp} helps us give a concise algorithm for optimally reporting. 

\begin{restatable}{corollary}{lemmasimpleinduciblesp}
\label{lemma: simple inducible sp}
For any extensive-form game $(T_{root},U_L,U_F)$, there exists a ``Y-shape'' inducible distribution $p^*\in \Delta(Z)$, such that $U_F(p^*)\geq U_F(p)$ for any inducible distribution $p$.
\end{restatable}

To simplify the following proofs, 
we present a non-recursive characterization of Y-shape inducible distributions here. 

\begin{restatable}{corollary}{corollarynonrecursivecharacterizationofyshapeinducibility}
\label{corollary: non-recursive characterization of y-shape inducibility}
A Y-shape realizable distribution $p$ is inducible if and only if 
\begin{enumerate}
    \item if $|\operatorname{Supp}(p)|=2$, and $U_L(z_1)\ne U_L(z_2)$, for $z_1,z_2 \in \operatorname{Supp}(p)$, then:
    \begin{enumerate}
        \item $z_1$ and $z_2$ have common ancestors $v\in \mathcal{P}^{-1}(F)$ and $w\in \child{v}$ such that 
        \begin{equation*}
            U_L(p)\geq  \min_{w'\in\child{v},w'\neq w}M_L(w');
        \end{equation*}
        \item $U_L(p) \geq M_L(u)$, $\forall u\in \mathcal{P}^{-1}(L)$ on the path from $root$ to $v$;
    \end{enumerate}
    \item if $|\operatorname{Supp}(p)|=1$ or $U_L(z_1) = U_L(z_2)$ for $z_1, z_2 \in \operatorname{Supp}(p)$, let $v$ be the least common ancestor of $z_1$ and $z_2$, then
    \begin{equation*}
    U_L(p) \geq M_L(u),~~\forall u\in \mathcal{P}^{-1}(L)~\text{on the path from}~root~\text{to}~v;
    \end{equation*}
\end{enumerate}
\end{restatable}

\subsection{Algorithm for Optimally Reporting}
\label{sec: algorithms under behavioral}

Now we give the algorithm for optimally reporting under behavioral commitment in Algorithm \ref{alg:OptimalCheatingBehavioral}. 

\begin{restatable}{theorem}{thmalgorithmforbehavioralcase}
\label{thm: algorithm for behavioral case}
Computing an inducible distribution $p^*\in\Delta(Z)$ and a 
payoff function $\tilde{U}_F^*$ such that 
\begin{itemize}
\item a 
strategy profile $(\delta_L,\pi_F)$ of $p^*$ is an SSE with behavioral commitment of game $(T_{root},U_L,\tilde{U}_F^*)$;
\item $U_F(p^*)\geq U_F(p)$ for all inducible distribution $p$;
\end{itemize}
can be solved in polynomial time.
\end{restatable}

\begin{proof}[Proof of \Cref{thm: algorithm for behavioral case}]
We prove that Algorithm \ref{alg:OptimalCheatingBehavioral} can find the optimal inducible distribution $p^*$ and the follower's payoff function $\tilde{U}_F^*$ in \Cref{thm: algorithm for behavioral case} in $O(|H|\cdot|Z|^2)$ time.

\textbf{Correctness.} The key observation of \Cref{lemma: simple inducible sp} shows that it suffices for us to find an optimal  Y-shape inducible distribution $p^*$, which dramatically simplifies our algorithm design.

Algorithm \ref{alg:OptimalCheatingBehavioral} enumerates all subtrees $T_v$ where $v\in H\cup Z$ and find all the inducible distributions that may be optimal, for all pairs $(z_1,z_2)\in E(v)$. The correctness then follows from \Cref{corollary: non-recursive characterization of y-shape inducibility}. 

\textbf{Complexity.} 
Firstly, for each $v\in H\cup Z$, we have $|E(v)|\leq |Z|^2$. So the complexity of line \ref{run mixture} would be $O(|H|\cdot|Z|^2)$.
We then analyze the complexity of the procedure $\mathtt{OptimalReport}$. For each $v\in H\cup Z$, the loop in line \ref{enumerate mixture} will run at most $O(|Z|^2)$ times, so the total complexity of the procedure $\mathtt{OptimalReport}$ would be $O(|H|\cdot|Z|^2)$.
Finally, after finding the optimal distribution $p^*$, by \Cref{coro: construct U_F in behavioral}, we can construct a payoff function that induces $p^*$,  $\tilde{U}_F^*$, in $O(|H|\cdot |Z|)$ time.

This completes the proof.
\end{proof}

\begin{algorithm}[!ht]
\caption{Optimally Reporting with Behavioral Commitment}
\label{alg:OptimalCheatingBehavioral}
\KwIn{
    An extensive-form game $(T_{root}, N, \mathcal{P},\{U_i\}_{i\in N})$.
}
\KwOut{
    A Y-shape inducible distribution $p^*$ and a 
    payoff 
    function $\tilde{U}_F^*$ such that (1) $U_F(p^*)\geq U_F(p)$ for all inducible distributions $p$; (2)an SSE with behavioral commitment of game $(T_{root},U_L,\tilde{U}_F^*)$ correspond to $p^*$.
}

\BlankLine

\SetKwProg{procedure}{Procedure}{}{}
\procedure{$\mathtt{Main}()$}{
	For each $v\in H\cup Z$, let $E(v)$ be the set of all pairs of leaf nodes $(z_1,z_2)$ with $z_1\neq z_2$ in subtree $T_v$, such that the least common ancestor of $z_1$ and $z_2$ belongs to $\mathcal{P}^{-1}(L)$.\label{run mixture}\\
	
	\tcp{This step will run in $O(|H|\cdot|Z|^2)$ time.}
	
	Initialize $p^*$ to be a 
	distribution 
	such that it corresponds to an SSE on game $(T_{root},U_L,U_F)$.
	
	Run $\mathtt{OptimalReport}(root,-\infty)$.\\
	\tcp{By $-\infty$ we mean that we initiate the second parameter to be sufficiently small. 
}
	Construct a follower's payoff function $\tilde{U}_F^*$ by \Cref{coro: construct U_F in behavioral}.\\
	\Return $\tilde{U}_F^*$ and $p^*$.
}
\procedure{$\mathtt{OptimalReport}(v,Lminval)$}{
    \If{$v$ is a leaf node}{
        If $U_L(v)\geq Lminval$, let $p(v)=1$.
            
        If $U_L(p)\geq U_L(p^*)$, then let $p^*=p$.
        
        \textbf{return}
    }
    
\For{$w\in\child{v}$}{
   $\mathtt{OptimalReport}(w,\max(Lminval,M_L(v)))$.\label{reduce}\\
   
\If{$\mathcal{P}(v)=F$}{
    Let $Threshold \leftarrow \max(\min_{w'\in\child{v},w'\neq w}M_L(w'),Lminval)$.\\
        \For{$(z_1,z_2)\in E(w)$}{\label{enumerate mixture}
            \textbf{if} $U_L(z_1)<Threshold$ and $U_L(z_2)<Threshold$ \textbf{continue}.
            
            Find $[a,b]\subseteq[0,1]$ s.t. $\forall \alpha\in[a,b],\alpha U_L(z_1)+(1-\alpha)U_L(z_2)\geq Threshold$.\\
            Find $\alpha\in\arg\max_{\alpha'\in[a,b]}(\alpha'(U_F(z_1),U_L(z_1))+(1-\alpha')(U_F(z_2),U_L(z_2))$.\\
            
            \tcp{$(a_1,b_1)< (a_2,b_2)$ if and only if \textbf{either} $a_1<a_2$ \textbf{or} $a_1=a_2$ and $b_1<b_2$.}
            
            Let $p(z_1)=\alpha$ and $p(z_2)=1-\alpha$.
            
            If $U_L(p)\geq U_L(p^*)$, then let $p^*=p$. \label{update answer}
        }
    }
    }
}
\end{algorithm}

\subsection{Follower Can Get No Less Actual Utility under Behavioral Commitment}
\label{sec: higher utility than pure}

Though the characterization conditions of inducible distributions become much more complicated, however, as we show in \Cref{proposition: higher utility}, via optimally inducing an inducible distribution, the follower can get no less actual utility under behavioral commitment than that under pure commitment. 

\begin{restatable}{proposition}{propositionhigherutility}
\label{proposition: higher utility}
Let $P(\Pi_L\times \Pi_F)$ ($P(\Delta_L\times \Pi_F)$) be the set of all inducible leaf nodes (inducible 
distributions, respectively). Then
\begin{equation*}
    \max_{p\in P(\Delta_L\times \Pi_F)} U_F(p) \geq \max_{z\in P(\Pi_L\times \Pi_F)} U_F(z). 
\end{equation*}
\end{restatable}

\section{Raise the Bar: Induce the Unique Outcome}
\label{sec: strong inducibility}

So far, we have characterized all the inducible 
distributions on leaf nodes, i.e., the possible outcomes that can be led to by an
induced SSE. 
However, while all SSEs yield the same utilities for the leader, there is still the equilibrium selection problem for the follower. An example is shown in \Cref{example: different utility and non-existence of optimal strong inducibility} that different SSEs yield dramatically different utilities for the follower. 

Thus in this section, we proceed to a more restricted condition for strategic reporting: When can a 
distribution be the only outcome that all SSEs of an induced game lead to? 

Formally, we define 
strong inducibility,  which is orginally proposed in \citet{birmpas2021optimally}.

\begin{definition}[Strong Inducibility]
A distribution $p\in \Delta(Z)$ is \textbf{strongly inducible with pure (behavioral) commitment} in game $(T_{root}, U_L,U_F)$ if there exists a payoff function $\tilde{U}_F$, such that all SSEs with pure (behavioral, respectively) commitments of game $(T_{root},U_L,\tilde{U}_F)$ correspond to $p$. 
\end{definition}

We characterize all the strongly inducible outcomes and give related algorithmic results with different commitments in \Cref{sec: characteriation of strong inducibility}. 
We then compare the optimal utilities the follower can get in \textbf{one} game under inducible and strongly inducible distributions, respectively. We give a characterization of games where these two values are equal
in \Cref{sec: efficiency of strong inducibility}. 

\subsection{Characterization and Near Optimality}
\label{sec: characteriation of strong inducibility}

We first give the characterization of strong inducibility with pure commitment.

\begin{restatable}{theorem}{thmstronginducibilityunderpure}
\label{thm: strong inducibility under pure}
Given leaf node $z\in Z$, let $v\in\child{root}$ be
the ancestor of $z$, 
then $z$ is strongly inducible with pure commitment if and only if
\begin{enumerate}
    \item if $\mathcal{P}(root)=L$, then all of the following conditions are met:
    \begin{enumerate}
        \item $U_L(z)> \max_{v^\prime \in \child{root}, v^\prime\ne v} M_L(v^\prime)$;
        \item $z$ is strongly inducible at $T_v$.
    \end{enumerate}
    \item if $\mathcal{P}(root)=F$, then at least one of following two conditions are met: 
    \begin{enumerate}
        \item $U_L(z)> \min_{v^\prime \in \child{root}, v^\prime\ne v} M_L(v^\prime)$;
        \item $z$ is strongly inducible at $T_v$. 
    \end{enumerate}
\end{enumerate}
\end{restatable}

As for the behavioral commitment case, we here give a sufficient and necessary condition for all  Y-shape strongly inducible distributions. Actually, this suffices for us to give algorithms and find the characterization in the next section.

\begin{restatable}{theorem}{thmstronginducibilityunderbehavioral}
\label{thm: strong inducibility under behavioral}
A ``Y-shape'' realizable distribution on leaf nodes $p\in \Delta(Z)$ is strongly inducible if and only if 
\begin{enumerate}
    \item if $\mathcal{P}(root)=L$, then $|\operatorname{Supp}(p,root)
        |=1$. Let $\operatorname{Supp}(p,root)
        =\{v\}$, then all of the following conditions are met:
    \begin{enumerate}
        \item $U_L(p)> \max_{v^\prime \in \child{root}, v^\prime\ne v} M_L(v^\prime)$;
        \item $p|_{v}$ is strongly inducible at 
        $T_v$. 
    \end{enumerate}
    \item if $\mathcal{P}(root)=F$,
     let $\operatorname{Supp}(p,root)
        =\{v\}$, then $U_L(z_1)\ne U_L(z_2)$
        if $\operatorname{Supp}(p)= \{z_1,z_2\}$; and at least one of following two conditions are met: 
    \begin{enumerate}
        \item $U_L(p)> \min_{v^\prime \in \child{root}, v^\prime\ne v} M_L(v^\prime)$;
        \item $p|_{v}$ is strongly inducible at $T_v$. 
    \end{enumerate}
\end{enumerate}
\end{restatable}

The strict inequalities in \Cref{thm: strong inducibility under behavioral} show that different from the case of  inducibility, strongly inducible distributions may not even exist: e.g., games that the leader has a constant payoff function. Furthermore, if strongly inducible distributions do exist, there may not be an optimal one (see an example in \Cref{example: different utility and non-existence of optimal strong inducibility}); even if an optimal strongly inducible distribution exists, the follower may get far less utility than he does under optimal inducible distributions (see an example in \Cref{example: different utility between optimal inducibility and optimal strong inducibility}).

We show it is polynomial-time tractable to decide under which aforementioned case a game is, and find an (near-)optimal solution if it exists. 

\begin{restatable}{theorem}{thmalgorithmforstronglyinducibleunderpure}
\label{thm: algorithm for strongly inducible under pure}
It is polynomial-time tractable to

(1) decide if a strongly inducible leaf node exists;

(2) if so, find an optimal strongly inducible leaf node and construct a 
payoff function that induces it. 

\end{restatable}

Denote the set of all strongly inducible distributions as $ SP(\Delta_L\times \Pi_F)$, we have, 
\begin{restatable}{theorem}{thmalgorithmforstronglyinducibleunderbehavioral}
\label{thm: algorithm for strongly inducible under behavioral}
It is polynomial-time tractable to 

(1) decide if a strongly inducible 
distribution exists;

(2) if so, decide if an optimal strongly inducible 
distribution exists;

(3) if so, find an optimal strongly inducible 
distribution; if not, for any $\epsilon>0$, find an $\epsilon$-optimal strongly inducible 
distribution $p^*$ such that
\begin{equation*}
   U_F(p^*)\geq \sup_{p\in SP(\Delta_L\times \Pi_F)}U_F(p)-\epsilon. 
\end{equation*}
In both cases, construct a payoff function that induces the distribution in polynomial time. 
\end{restatable}

Similarly, we can show that the follower can get no less utility by optimally strategically reporting under behavioral commitment case, than under pure commitment case. Denote the set of all strongly inducible leaf nodes as $SP(\Pi_L\times\Pi_F)$, and we have, 

\begin{restatable}{proposition}{propositionhigherutilityunderstronglyinducibility}
\label{proposition: higher utility under strongly inducibility}
$\sup_{p\in SP(\Delta_L\times \Pi_F)} U_F(p) \geq \sup_{z\in SP(\Pi_L\times \Pi_F)} U_F(z)$.
\end{restatable}

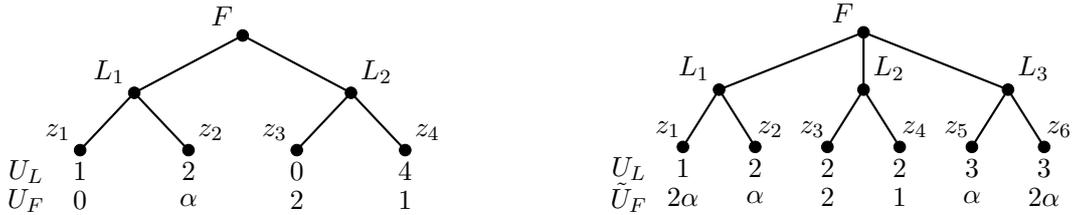
\begin{figure}[!htbp]
\captionsetup[subfigure]{}
\subcaptionbox{A game that 1. any payoff function that induces a desired distribution always yields another SSE with far less follower's utility; 2. optimal strongly inducible distributions do not exist. \label{example: different utility and non-existence of optimal strong inducibility}}[.49\textwidth]{%
\centering
\begin{tikzpicture}[scale=0.95, every node/.style={transform shape}]
\node[above left] at (3,1.6) {$F$};
\node[above left] at (1.5,0.8) {$L_1$};
\node[above right] at (4.5,0.8) {$L_2$};
\node[above left] at (0.75,0) {$z_1$};
\node[above right] at (2.25,0) {$z_2$};
\node[above left] at (3.75,0) {$z_3$};
\node[above right] at (5.25,0) {$z_4$};
\node at (0,-0.3) {$U_L$};
\node at (0,-0.7) {$U_F$};
\node at (0.75,-0.3) {$1$};
\node at (0.75,-0.7) {$0$};
\node at (2.25,-0.3) {$2$};
\node at (2.25,-0.7) {$\alpha$};
\node at (3.75,-0.3) {$0$};
\node at (3.75,-0.7) {$2$};
\node at (5.25,-0.3) {$4$};
\node at (5.25,-0.7) {$1$};
\filldraw (3,1.6) circle (.08)
          (1.5,0.8) circle (.08)
          (4.5,0.8) circle (.08)
          (0.75,0) circle (.08)
          (2.25,0) circle (.08)
          (3.75,0) circle (.08)
          (5.25,0) circle (.08);
\draw[thick] (3,1.6)--(1.5,0.8)
      (3,1.6)--(4.5,0.8)
      (1.5,0.8)--(0.75,0)
      (1.5,0.8)--(2.25,0)
      (4.5,0.8)--(3.75,0)
      (4.5,0.8)--(5.25,0);
\end{tikzpicture}
}
\hfill
\subcaptionbox{A game that optimal strongly inducible distributions yield far less follower's utility optimal inducible distributions do.\label{example: different utility between optimal inducibility and optimal strong inducibility}}[.49\textwidth]{
\begin{tikzpicture}[scale=0.95, every node/.style={transform shape}]
\centering
\node[above left] at (3,1.6) {$F$};
\node[above left] at (1,0.8) {$L_1$};
\node[above right] at (3,0.8) {$L_2$};
\node[above right] at (5,0.8) {$L_3$};
\node[above left] at (0.6,0) {$z_1$};
\node[above right] at (1.4,0) {$z_2$};
\node[above left] at (2.6,0) {$z_3$};
\node[above right] at (3.4,0) {$z_4$};
\node[above left] at (4.6,0) {$z_5$};
\node[above right] at (5.4,0) {$z_6$};
\node at (-0.25,-0.3) {$U_L$};
\node at (-0.25,-0.7) {$\tilde{U}_F$};
\node at (0.5,-0.3) {$1$};
\node at (0.5,-0.7) {$2\alpha$};
\node at (1.5,-0.3) {$2$};
\node at (1.5,-0.7) {$\alpha$};
\node at (2.5,-0.3) {$2$};
\node at (2.5,-0.7) {$2$};
\node at (3.5,-0.3) {$2$};
\node at (3.5,-0.7) {$1$};
\node at (4.5,-0.3) {$3$};
\node at (4.5,-0.7) {$\alpha$};
\node at (5.5,-0.3) {$3$};
\node at (5.5,-0.7) {$2\alpha$};
\filldraw (3,1.6) circle (.08)
          (1,0.8) circle (.08)
          (3,0.8) circle (.08)
          (5,0.8) circle (.08)
          (0.5,0) circle (.08)
          (1.5,0) circle (.08)
          (2.5,0) circle (.08)
          (3.5,0) circle (.08)
          (4.5,0) circle (.08)
          (5.5,0) circle (.08);
\draw[thick] (3,1.6)--(1,0.8)
      (3,1.6)--(3,0.8)
      (3,1.6)--(5,0.8)
      (1,0.8)--(0.5,0)
      (1,0.8)--(1.5,0)
      (3,0.8)--(2.5,0)
      (3,0.8)--(3.5,0)
      (5,0.8)--(4.5,0)
      (5,0.8)--(5.5,0);
\end{tikzpicture}}
\caption{Examples for \Cref{sec: strong inducibility}. $\alpha\in (-\infty,-1]$ Fig. \ref{example: different utility and non-existence of optimal strong inducibility}: The optimal inducible distribution of this game $p$ satisfies  $p(z_3)=p(z_4)=\frac{1}{2}$, yielding a utility $2=M_{L}(L_1)$ for the leader. The leader can always choose to play maximin strategy on $T_F$ to get the same utility, whichever payoff function the follower reports, while the follower get utility $\alpha \ll U_F(p)=1.5$. Also $p$ is not strongly inducible, the follower can get utility arbitrarily close to $1.5$ under strongly inducible distributions, but cannot get exactly $1.5$. Fig.~\ref{example: different utility between optimal inducibility and optimal strong inducibility}: The optimal inducible distribution of this game is $p(z_3)=p(z_4)=\frac{1}{2}$, yielding a utility $2=M_{L}(L_1)$ for the leader and $1.5$ for the follower. However, the follower can only get $\alpha$ under optimal strongly inducible distributions, which is far less than $1.5$.  }
\label{example: different strong inducibility}
\end{figure}

\subsection{Characterization of When Follower Gains Arbitrarily The Same Optimally in the Two Senses of Inducibility}
\label{sec: efficiency of strong inducibility}

Note that the follower's optimal utility through strategic reporting can be dramatically different between the case of inducibility  and of strong inducibility. 
We are interested that under what kind of games can the follower's optimal values in these two cases are arbitrarily close. 
That is, 
denote an optimal inducible distribution as $p^*\in \argmax_{p\in P(\Delta_L\times \Pi_F)}U_F(p)$. 
We wish to find a sufficient and necessary condition for a game to satisfy the following ``\textbf{Utility Supremum Equivalence}''(USE) property ($P$) with behavioral commitment:
\begin{equation}
\tag{$P$}
\label{eqn: near optimality}
    \sup_{p\in SP(\Delta_L\times \Pi_F)}\frac{U_F(p)}{U_F(p^*)} = 1
\end{equation}

 We note that, 
 due to the finiteness of all (strongly) inducible leaf nodes with pure commitment, a game satisfies USE with pure commitment if and only if one of the optimal inducible leaf nodes is also strongly inducible. 

As for the behaiovral commitment case, we first provide two conditions. \Cref{thm: full characterization of near optimality} shows that these two conditions fully characterize property \ref{eqn: near optimality}. 

\begin{condition}
\label{condition 1}
There exists an optimal ``Y-shape'' inducible distribution $p^*$, such that, $|\operatorname{Supp}(p^*)|=2$ and $U_L(z_1) \ne U_L(z_2)$ where $\operatorname{Supp}(p^*)=\{z_1,z_2\}$. 
\end{condition}

\begin{condition}
\label{condition 2}
One of the following conditions are met for some optimal ``Y-shape'' inducible distribution $p^*$:  
\begin{enumerate}
    \item $p^*$ is strongly inducible;
    \item For some $z^* \in Z$ be the leaf node that $p^*(z^*)>0$, there exists a node $v$ on the path from $root$ to $z^*$ such that
        \begin{enumerate}
            \item $\mathcal{P}(v)=F$, and $U_L(z^*)\geq \min_{w\in \child{v}, w\ne w(v,z^*)}M_L(w)$, where $w(v,z^*) \in \child{v}$ is $z$'s ancestor; 
            \item $U_L(z_1) > U_L(z^*)$ for some $(z_1,z^*) \in E(v)$, i.e., $z_1\in Z_v$ is a leaf node that its least common ancestor with $z^*$ belongs to $\mathcal{P}^{-1}(L)$.
        \end{enumerate}
\end{enumerate}
\end{condition}

\begin{restatable}{theorem}{thmfullcharacterizationofnearoptimality}
\label{thm: full characterization of near optimality}
Game $(T_{root},U_L,U_F)$ satisfies property \ref{eqn: near optimality} if and only if it satisfies \Cref{condition 1} or \Cref{condition 2}. 
\end{restatable}

 One important observation for proving \Cref{thm: full characterization of near optimality} is that: recalling that we can  find an optimal inducible distribution that is Y-shape in polynomial time, 
 and characterize all the Y-shape strongly inducible distributions, denote the set of all Y-shape strongly inducible distributions as $SYP(\Delta_L \times \Pi_F)$, and we consider another property ($P^\prime$): 
 \begin{equation}
\tag{$P^\prime$}
\label{eqn: y-shape near optimality}
    \sup_{p\in SYP(\Delta_L\times \Pi_F)}\frac{U_F(p)}{U_F(p^*)} = 1
\end{equation}

\Cref{lemma: equivalence of two near optimality} shows that these two properties are equivalent. 
 
\begin{restatable}{lemma}{lemmaequivalenceoftwonearoptimality}
\label{lemma: equivalence of two near optimality}
Game $(T_{root},U_L,U_F)$ satisfies property \ref{eqn: near optimality} if and only if it satisfies property \ref{eqn: y-shape near optimality}. 

\end{restatable}

\section{Conclusion and Future Works}
\label{sec: conclusion}

We study how the follower strategically reports his payoff function to gain optimal utility, when facing a learning leader who constantly interact with the follower and collect his utility data. We characterize all the game outcomes that can be successfully induced under different settings, and give efficient algorithms to find the optimal way of strategically reporting. We completely resolve this problem in extensive-form games with perfect information. 
One of the future work is to study such follower's strategic behavior in the computation of other solution concepts or in other settings. 
Though in some game settings, e.g., the computation of 
Stackelberg
equilibria in extensive-form games with imperfect information is NP-hard ~\citep{letchford2010computing},
and the approximation of Nash equilibria in normal-form games is PPAD-complete~\citep{daskalakis2009the, chen2009settling}, it is still interesting to study such follower's strategic behavior in such settings. 

Another future work would be to consider how to counteract such manipulation. Actually, when the follower may misreport to exercise his right of privacy, we assume he is in a higher cognitive hierarchy than the leader. 
Since the follower only needs to best-respond to the leader's commitment, it is not necessary for the leader to hide her own private payoff information, if he is not aware of the follower's strategic behavior.  
How about these two players being in the same hierarchy, or the leader's hierarchy is higher than the follower's? It is worth mentioning that if we also allow the leader to report a fake payoff function, i.e., think about a new game built from the original 
Stackelberg
game where both players' strategies are which payoff functions to announce. Then it is an NE that the leader reports her true function and the follower optimally strategically reports according to our results. A complete equilibrium analysis might help us find a better way to come up with a countermeasure.

\clearpage

\appendix

\section{Omitted Proofs from Section \ref{sec: properties for pure case}}
\label{app1}

\subsection{Proof of \Cref{lemma: min utility in pure commit}}
\label{subsec: lemmaminutilityinpurecommit}
\lemmaminutilityinpurecommit*

Let $U_L(\pi_L,\pi_F)|_{z} = U_L(z)$ for $z\in Z$ to simplify the notation. To make the proof clear, we first prove two lemmas. 

\begin{lemma}
\label{lemma: strategy form of max and min}
\begin{enumerate}
    \item 
    \begin{enumerate}
        \item Given the leader's pure strategy $\pi_L$, the following follower's pure strategies minimize the leader's utility,
        \begin{equation}
\label{eqn: follower's min strategy}
    \pi_F(v)\in {\arg\min}_{w\in \child{v}} U_L(\pi_L, \pi_F)|_{w},~~~~~~~\forall v\in \mathcal{P}^{-1}(F);
\end{equation}
    \item Assuming that the follower always minimizes the leader's utility, using strategies defined by \Cref{eqn: follower's min strategy}, the following leader's pure strategies  maximize her own utility,
        \begin{equation}
\label{eqn: leader's max strategy}
    \pi_L(v)\in {\arg\max}_{w\in \child{v}} U_L(\pi_L, \pi_F)|_{w}~~~~~~~\forall v\in \mathcal{P}^{-1}(L).
\end{equation}
    \end{enumerate}
    \item 
    \begin{enumerate}
        \item Given the leader's pure strategy $\pi_L$, the following follower's pure strategies maximize his own utility, 
        \begin{equation}
    \label{eqn: leader's min strategy}
    \pi_F(v) \in 
{\arg\max}_{w\in \child{v}} U_F(\pi_L, \pi_F)|_{w}, ~~~~~~~\forall v\in \mathcal{P}^{-1}(F);
\end{equation}
    \item Assuming that the follower always maximizes his own utility, using strategies defined by \Cref{eqn: leader's min strategy}, the following leader's pure strategies  minimize the follower's utility,
        \begin{equation}
\label{eqn: follower's max strategy}
    \pi_L(v)\in {\arg\min}_{w\in \child{v}} U_F(\pi_L, \pi_F)|_{w},~~~~~~~\forall v\in \mathcal{P}^{-1}(L).
\end{equation}
    \end{enumerate}
\end{enumerate}
\end{lemma}
\begin{proof}
We only prove (1) here, the proof of (2) is analogous. 

\textbf{First we prove (1)(a).} Intuitively, starting at any strategy ${\pi_F'}$ of the follower, changing any action at any node $v\in \mathcal{P}^{-1}(F)$, to the one leading to the least value of the leader in the subgame $T_v$ always decreases the leader's utility. 

Given the leader's strategy $\pi_L\in \Pi_L$, we will prove that $U_L(\pi_L,\pi^\prime_F)|_{v} \geq U_L(\pi_L,\pi_F)|_{v}$ for any $\pi^\prime_F \in \Pi_F$ and any $v\in H\cup Z$ by induction. 

\textbf{Inductive Base:} When $v\in Z$, $U_L(\pi_L,\pi_F^\prime)|_{v}=U_L(v)=U_L(\pi_L,\pi_F)|_{v}$, thus $U_L(\pi_L,\pi^\prime_F)|_{v} \geq U_L(\pi_L,\pi_F)|_{v}$ holds. 

\textbf{Inductive Step:} When $v\in H$, assume the inductive hypothesis holds, i.e., $U_L(\pi_L,\pi^\prime_F)|_{{w^\prime}}\geq U_L(\pi_L,\pi_F)|_{{w^\prime}}$ for any $w^\prime\in \child{v}$. 

If $\mathcal{P}(v)=L$, suppose $\pi_L(v)=w$, then $U_L(\pi_L,\pi^\prime_F)|_{v} =U_L(\pi_L,\pi^\prime_F)|_{w}\geq U_L(\pi_L,\pi_F)|_{w} = U_L(\pi_L,\pi_F)|_{v}$. 

If $\mathcal{P}(v)=F$, suppose $\pi^\prime_F(v)=w$, then
$U_L(\pi_L,\pi^\prime_F)|_{v} = U_L(\pi_L,\pi^\prime_F)|_{w}\geq U_L(\pi_L,\pi_F)|_{w} \geq \min_{w^\prime\in\child{v}} U_L(\pi_L,\pi_F)|_{{w^\prime}} = U_L(\pi_L,\pi_F)|_{v}$. 

\textbf{Then we prove (1)(b). } Let $\pi_L \in \Pi_L$ be the leader's strategy defined by \Cref{eqn: leader's max strategy}, and $\pi_F \in \Pi_F$ the corresponding follower's strategy defined by \Cref{eqn: follower's min strategy}, we prove for any leader's strategy $\pi^\prime_L\in \Pi_L$ and the corresponding follower's strategy $\pi^\prime_F \in \Pi_F$ defined by \Cref{eqn: follower's min strategy}, $U_L(\pi_L,\pi_F)|_{v}\geq U_L(\pi^\prime_L,\pi^\prime_L)|_{v}$ for any $v\in H\cup Z$. 

\textbf{Inductive Base: } When $v\in Z$, $U_L(\pi_L,\pi_F)|_{v} = U_L(v) = U_L(\pi^\prime_L,\pi^\prime_F)|_{v}$.

\textbf{Inductive Step: }When $v\in H$, assume the inductive hypothesis holds, i.e., $U_L(\pi_L,\pi_F)|_{{w^\prime}}\geq U_L(\pi^\prime_L,\pi^\prime_F)|_{{w^\prime}}$ for any $w^\prime\in \child{v}$.  

If $\mathcal{P}(v)=L$, suppose $\pi^\prime_L(v) = w$, then 
\begin{equation*}
    U_L(\pi_L,\pi_F)|_{v} = \max_{w^\prime \in \child{v}}U_L(\pi_L,\pi_F)|_{{w^\prime}} \geq U_L(\pi_L,\pi_F)|_{w} \geq U_L(\pi^\prime_L,\pi^\prime_F)|_{w} = U_L(\pi^\prime_L,\pi^\prime_F)|_{v}.
\end{equation*}

If $\mathcal{P}(v)=F$, since $U_L(\pi_L,\pi_F)|_{{w^\prime}}\geq U_L(\pi^\prime_L,\pi^\prime_F)|_{{w^\prime}}$ for any $w^\prime\in \child{v}$, we have 
\begin{equation*}
    U_L(\pi_L,\pi_F)|_{v} = \min_{w^\prime \in \child{v}}U_L(\pi_L,\pi_F)|_{{w^\prime}}\geq \min_{w^\prime \in \child{v}}U_L(\pi^\prime_L,\pi^\prime_F)|_{{w^\prime}} = U_L(\pi^\prime_L,\pi^\prime_F)|_{v}. 
\end{equation*}
\end{proof}

\begin{lemma}
\label{lemma: max-min strategy}
1. Strategy profile $(\pi_L,\pi_F)$ defined by  \Cref{eqn: follower's min strategy} and (\ref{eqn: leader's max strategy}) has the following form:
\begin{equation}
\label{eqn:leader's max strategy-close form}
        \pi_L(v)\in{\arg\max}_{w\in \child{v}} M_L(w)
\end{equation}
\begin{equation}
\label{eqn:follower's min strategy-close form}
        \pi_F(v)\in{\arg\min}_{w\in\child{v}} M_L(w)
\end{equation}
and $U_L(\pi_L,\pi_F)=M_L(root)$.

2. Strategy profile $(\pi_L,\pi_F)$ defined by \Cref{eqn: leader's min strategy} and (\ref{eqn: follower's max strategy}) has the following form:
\begin{equation}
        \pi_L(v)\in{\arg\min}_{w\in \child{v}} M_F(w)
\end{equation}
\begin{equation}
        \pi_F(v)\in{\arg\max}_{w\in\child{v}} M_F(w)
\end{equation}
and $U_F(\pi_L,\pi_F)=M_F(root)$
\end{lemma}

\begin{proof}
We only prove the first statement here, the proof of the second statement is similar. Besides the form of the strategy profiles, we will also prove that under that strategy profile $(\pi_L,\pi_F)$ defined by \Cref{eqn: follower's min strategy} and (\ref{eqn: leader's max strategy}), $U_L(\pi_L,\pi_F)|_{v} = M_L(v)$ for all $v\in H\cup Z$. Thus $U_L(\pi_L,\pi_F)= M_L(root)$.  

\textbf{Inductive Base:} When $v\in Z$, $U_L(\pi_L,\pi_F)|_{v} = U_L(v) = M_L(v)$. The action space is empty here, the statements hold. 

\textbf{Inductive Step:} When $v\in H$, assume the inductive hypothesis holds and $(\pi_L,\pi_F)$ in \Cref{eqn: follower's min strategy} and (\ref{eqn: leader's max strategy}) have been defined in its descendent nodes.  

If $\mathcal{P}(v) = L$, $\pi_L(v) \in \argmax_{w\in\child{v}} U_L(\pi_L,\pi_F)|_{w} = {\arg\max}_{w\in\child{v}} M_L(w)$ by the inductive hypothesis, and  $U_L(\pi_L,\pi_F)|_{v} = \max_{w\in\child{v}}U_L(\pi_L,\pi_F)|_{w} = \max_{w\in\child{v}}M_L(w) = M_L(v)$. 

If $\mathcal{P}(v) = F$, $\pi_F(v) \in {\arg\min}_{w\in\child{v}} U_L(\pi_L,\pi_F)|_{w} = {\arg\min}_{w\in\child{v}} M_L(w)$ by the inductive hypothesis, and $U_L(\pi_L,\pi_F)|_{v} = \min_{w\in\child{v}}U_L(\pi_L,\pi_F)|_{w} = \min_{w\in\child{v}}M_L(w) = M_L(v)$. 
\end{proof}
\begin{proof}[Proof of \Cref{lemma: min utility in pure commit}]

Since there is an SSE leading to $z$, we have
\begin{equation}
\label{eqn: max-min full}
 U_L(z) = \max_{{\pi_L'}\in\Pi_L} U_L({\pi_L'},\pi^\prime_F \in \BR({\pi_L'})) \geq \max_{{\pi_L'}\in \Pi_L}\min_{{\pi_F'} \in \Pi_F} U_L({\pi_L'},{\pi_F'}). 
\end{equation}
By \Cref{lemma: strategy form of max and min} and \Cref{lemma: max-min strategy}, the right hand side of \Cref{eqn: max-min full} equals $M_L(root)$.

\begin{equation}
\label{eqn:min-max full}
     U_F(z)=\max_{\pi^\prime_F\in \Pi_F}U_F(\pi_L,\pi^\prime_F) \geq \min_{{\pi_L'}\in \Pi_L}\max_{\pi^\prime_F\in \Pi_F}U_F({\pi_L'},\pi^\prime_F).
\end{equation}
By \Cref{lemma: strategy form of max and min} and \Cref{lemma: max-min strategy}, the right hand side of \Cref{eqn:min-max full} equals $M_F(root)$. 

\end{proof}

\subsection{Proof of \Cref{lemma: zero-sum in pure commit}}
\label{subsec: lemmazerosuminpurecommitment}

\lemmazerosuminpurecommitment*

\begin{proof}
Since $U_F =-U_L$, we have $M_F(root)=-M_L(root)$. Let $(\pi^*_L,\pi^*_F)$ be an SSE with pure commitment of 
$(T_{root},U_L,-U_L)$, we have \begin{equation*}
    \begin{aligned}
    U_L(\pi^*_L,\pi^*_F)&=\mathop{\max}_{\pi_L\in \Pi_L,\pi_F\in \BR(\pi_L)}U_L(\pi_L,\pi_F) \geq M_L(root),\\
    U_F(\pi^*_L,\pi^*_F)&=-U_L(\pi^*_L,\pi^*_F) \geq M_F(root)=-M_L(root),
    \end{aligned}
\end{equation*}
which means
\begin{equation*}
\begin{aligned}
& M_L(root) \leq \mathop{\max}_{\pi_L\in \Pi_L,\pi_F\in \BR(\pi_L)}U_L(\pi_L,\pi_F) \leq M_L(root)\\ \Rightarrow & \mathop{\max}_{\pi_L\in \Pi_L,\pi_F\in \BR(\pi_L)}U_L(\pi_L,\pi_F) = M_L(root).
\end{aligned}
\end{equation*}
\end{proof}

\subsection{Proof of \Cref{lemma: pure commit}}
\label{subsec: lemmapurecommit}
\lemmapurecommit*
\begin{proof}
For \textbf{the sufficiency}, suppose that $z\in Z$ satisfies $U_F(z)\geq M_F(v)$ for all nodes $v$ on the path from $root$ to $z$. For node $v\in H$ that is on the path from $root$ to $z$, we use $w(v,z)\in\child{v}$ to denote the child of $v$ that is on the path from $root$ to $z$. Consider the following commitment of the leader
\begin{equation*}
\pi_L(v) \left\{ \begin{array}{ll}
=w(v,z),& {v\in \mathcal{P}^{-1}(L) \textrm{ is on the path from }root\textrm{ to }z};\\
\in {\arg\min}_{w\in \child{v}}M_F(w),& \text{otherwise},
\end{array}
\right.
\end{equation*}
and the following strategy of the follower
\begin{equation*}
\pi_F(v) \left\{ \begin{array}{ll}
=w(v,z),& {v\in \mathcal{P}^{-1}(F)\textrm{ is on the path from }root\textrm{ to }z};\\
\in {\arg\max}_{w\in \child{v}}M_F(w),& \text{otherwise}.
\end{array}
\right.
\end{equation*}

Notice that $(\pi_L,\pi_F)$ will lead to the leaf node $z$, so we only need to prove $\pi_F\in\BR(\pi_L)$. This is mainly based on the following observation: for each $v\in \mathcal{P}^{-1}(F)$ that is not on the path from $root$ to $z$, we have \begin{equation}
\label{eqn: not on the path}
    \max_{\pi_F'\in\Pi_F}U_F(\pi_L,\pi_F')|_{v}=M_F(v),
\end{equation}
and for each $v\in \mathcal{P}^{-1}(F)$ that is on the path from $root$ to $z$, we have 
\begin{equation}
\label{eqn: on the path}
    \max_{\pi_F'\in\Pi_F}U_F(\pi_L,\pi_F')|_{v}=U_F(z).
\end{equation}

The proof of Equation \ref{eqn: not on the path} can be directly derived by \Cref{lemma: max-min strategy}. We then prove Equation \ref{eqn: on the path} for every $v\in \mathcal{P}^{-1}(F)$ that is on the path from $root$ to $z$. Intuitively, we can observe the Equation \ref{eqn: on the path} by a simple induction over the tree structure.
If $\pi^\prime_F(v)=w(v,z)$, then $U_F(\pi_L,\pi^\prime_F)|_{v}\leq U_F(z)$ by the inductive hypothesis. If $\pi^\prime_F(v)=w'\neq w(v,z)$, then $U_F(\pi_L,\pi^\prime_F)|_{v}=M_F(w')$. By the condition $U_F(z)\geq M_F(v)=\max_{w'\in\child{v}}M_F(w')$ and thus Equation \ref{eqn: on the path} holds on $v$. 
Thus
\begin{equation*}
    \max_{\pi_F'\in\Pi_F}U_F(\pi_L,\pi_F')=U_F(z) = U_F(\pi_L,\pi_F). 
\end{equation*}

For \textbf{the necessity} part, suppose that 
$(\pi_L,\pi_F\in\BR(\pi_L))$ leads to $z\in Z$. Suppose, for the sake of contradiction, that there is a node $v\in \mathcal{P}^{-1}(F)$ on the path from $root$ to $z$ such that $U_F(z)<M_F(v)$. Then the follower can construct a new strategy $\pi_F^*$, where $\pi_F^*(v)= w \in \argmax_{w^\prime\in \child{v}}M_F(w^\prime)$, $\pi_F^*\in \BR(\pi_L)|_{w}$ and $\pi_F^*(v^\prime)=\pi_F(v^\prime)$ for other undefined nodes $v'\in \mathcal{P}^{-1}(F)$. Then by \Cref{eqn:min-max full}, $U_F(\pi_L,\pi_F^*)|_{{v}}= U_F(\pi_L,\pi_F^*)|_{w}\geq M_F(w) = M_F(v)$, which means $U_F(\pi_L,\pi_F^*)|_{v}\geq M_F(v) > U_F(\pi_L,\pi_F)|_{v}$, contradicting to our assumption. 
\end{proof}

\section{Omitted Proofs from Section \ref{sec: algorithms under pure}}

\subsection{Proof of \Cref{thm: algorithm for pure case}}
\label{appendix: thm: algorithm for pure case}
\thmalgorithmforpurecase*

\begin{proof}
We prove that Algorithm \ref{alg:OptimalCheatingPure} is a linear-time algorithm to find the $z^*$ and $\tilde{U}_F^*$ in \Cref{thm: algorithm for pure case}.

\textbf{Correctness.} By \Cref{thm: cheat in pure commit}, 
$SF$ is the set of all inducible leaf nodes. So $z^*$ satisfies $U_F(z^*)\geq U_F(z)$ for all inducible leaf nodes $z$. By the proof of \Cref{thm: cheat in pure commit}, we know that there is an SSE with pure commitment of game $(T_{root},U_L,\tilde{U}_F^*)$ leading to $z^*$. So the correctness follows.

\textbf{Complexity.} We note that $M_L(root)$ can be computed in $O(|H|+|Z|)$ time.
Calculating the set $SF$ and picking the leaf node $z^*$ would run in $O(|Z|)$ time.
By \Cref{coro:construct U_F in pure},
we can construct the follower's payoff function in $O(|H|+|Z|)$ time. So the complexity of Algorithm \ref{alg:OptimalCheatingPure} is $\Theta(|H|+|Z|)$.
\end{proof}

\begin{algorithm}[!ht]
	\caption{Optimally Reporting with Pure Commitment}
	\label{alg:OptimalCheatingPure}
	\KwIn{
	    An extensive-form game $(T_{root}, N=\{L,F\}, \mathcal{P},\{U_L,U_F\})$.
	}
	\KwOut{
	    An inducible leaf node $z^*$ and a 
	    payoff function $\tilde{U}_F^*$ satisfying 
	    (1) 
	    a 
	    strategy profile  of $z^*$ is an SSE with pure commitment of game $(T_{root},U_L,\tilde{U}_F^*)$,
	    (2) $U_F(z^*)\geq U_F(z)$ for all inducible leaf nodes $z$.
	}
	
	\BlankLine
	Calculate a set of leaf nodes $SF=\{z\in Z|U_L(z)\geq M_L(root)\}$.\\
	Pick a leaf node $z^*\in SF$ such that $z^*\in\arg\max_{z\in SF}U_F(z)$.\\
	Construct a follower' payoff function $\tilde{U}_F^*$ via \Cref{coro:construct U_F in pure}. \\
	\Return $\tilde{U}_F^*$ and $z^*$.
\end{algorithm}

\section{Omitted Proofs from Section \ref{sec: properties for behavioral case}}

\subsection{Proof of \Cref{lemma: feasible}}
\label{subsec: lemmafeasible}
\lemmafeasible*
\begin{proof}
1. Since if not, suppose there exists ${\pi_F'}|_{w}$ such that $U_F(\delta_L,{\pi_F'})|_{w} > U_F(\delta_L,\pi_F)|_{w}$ for some $w\in\operatorname{Supp}((\delta_L,\pi_F)|_{v},v)$, then define the follower's strategy ${\pi_F'}|_{v}$ at $T_v$ to be equal to  ${\pi_F'}|_{w}$ at $T_w$ and 
$\pi_F|_{v}$ on other undefined nodes. Then $U_F(\delta_L,\pi^\prime_F)|_{v}>U_F(\delta_L,{\pi_F})|_{v}$,  contradicting to that $\pi_F \in \BR(\delta_L)|_{v}$.

2. Since for any follower's strategy ${\pi_F'}|_{v}$ at $T_v$, 
\begin{equation*}
    U_F(\delta_L,{\pi_F'})|_{w}\leq U_F(\delta_L,\pi_F)|_{w}, ~~~\forall w\in \operatorname{Supp}((\delta_L,\pi_F)|_{v},v),
\end{equation*}   
the convex combination of $\{U_F(\delta_L,{\pi_F'})|_{w}\}_{w\in\operatorname{Supp}((\delta_L,\pi_F)|_{v},v)}$ is at most the same combination of $\{U_F(\delta_L,\pi_F)|_{w}\}_{w\in\operatorname{Supp}((\delta_L,\pi_F)|_{v},v)}$. 

3. Suppose $(\delta_L,\pi_F)|_{v}(v) = \sum_{w^\prime\in \child{v}}\alpha_{w^\prime} w^\prime$ and, for the sake of contradiction, for some $w\in \operatorname{Supp}((\delta_L,\pi_F)|_{v},v)$, $\pi^\prime_F \in \argmax_ {\pi^{\prime\prime}_F \in \BR(\delta_L)|_{w}}U_L(\delta_L,\pi^{\prime\prime}_F)|_{w}$ satisfies $U_L(\delta_L,\pi^\prime_F)|_{w} > U_L(\delta_L,\pi_F)|_{w}$.
Extend $\pi^\prime_F|_{w}$ to $T_v$ by letting it equal $\pi_F|_{v}$ on other undefined nodes. By 1 and 2, $(\delta_L,\pi^\prime_F)|_{v}$ is feasible and $U_L(\delta_L,\pi^\prime_F)|_{v} = \sum_{w^\prime \in\child{v}}\alpha_{w^\prime} U_L(\delta_L,\pi^\prime_F)|_{{w^\prime}}> U_L(\delta_L,\pi_F)|_{v}$, leading to a contradiction. 
\end{proof}

\subsection{Proof of \Cref{coro: construct U_F in behavioral}}
\label{subsec: coroconstructuf}

\coroconstructuf*

\begin{proof}
Use backward induction to calculate $\tilde{U}_F$ by Equation \ref{calculate U_F1}, Equation \ref{calculate U_F2} and Equation \ref{calculate U_F3}. Note that for each $h\in H$, the Equation \ref{calculate U_F1} and Equation \ref{calculate U_F2} will use at most $|Z|$ steps to update $\tilde{U}_F$, so the total complexity is $O(|H|\cdot|Z|)$.
\end{proof}

\section{Omitted Proofs from Section \ref{sec: simple inducible strategy-behavioral}}

\subsection{Proof of \Cref{lemma: simple inducible sp}}

\lemmasimpleinduciblesp*

\begin{proof}

We prove by structural induction over the game tree 
that for any inducible distribution $p\in \Delta(Z)$ on game $(T_{root},U_L,U_F)$, there exists a Y-shape inducible distribution $p^*\in \Delta(Z)$ such that $U_F(p^*)\geq U_F(p)$ and $U_L(p^*)\geq U_L(p)$.

\textbf{Inductive Base:} When $root\in Z$, the only distribution $p\in\Delta(Z)$ satisfies $|\operatorname{Supp}(p)|=1<2$, and is inducible. 

\textbf{Inductive Step:} Suppose that $root \in H$ and $p\in \Delta(Z)$ is inducible.

\textbf{When $\mathcal{P}(root)=F$: } Since $p$ is realizable, let $v\in\child{root}$ be the only child of $root$ that $p(v)>0$. If $p|_{v}$ is inducible at $T_v$, 
then by the inductive hypothesis, we know that there exists a Y-shape inducible distribution $p^*|_{v}\in\Delta(Z_v)$ such that $U_F(p^*)|_{v}\geq U_F(p)|_{v}$ and $U_L(p^*)|_{v}\geq U_L(p)|_{v}$. Extend $p^*|_{v}$ to $T_{root}$ by letting $p^*(z)=0$ for each $z\in Z\setminus Z_v$,
then $p^*$ is Y-shape and inducible at
$T_{root}$, 
by \Cref{thm: cheat in behavioral commit} (2)(b). The inequalities of utilities still hold for $p^*$ and $p$. 

Now suppose $U_L(p)\geq \min_{v'\in\child{root},v'\neq v}M_L(v')$. If $p$ is not Y-shape, then we have $|\operatorname{Supp}(p)|>2$, which means $U_i(p)~(i\in N)$ are mixtures of utilities of more than two leaf nodes. Since there are only two players in the game, there exists a mixture of two leaf nodes in the set $\{z|z\in Z, p(z)>0\}$, say $p^*$, as a dominant strategy, i.e., $U_L(p^*)\geq U_L(p)$ and $U_F(p^*)\geq U_F(p)$. Since 
	\begin{equation*}
		U_L(p^*)\geq U_L(p)\geq \min_{v'\in\child{root},v'\neq v}M_L(v'),
	\end{equation*}
	by \Cref{thm: cheat in behavioral commit} (2)(a),  $p^*$ is inducible, and satisfies $U_F(p^*)\geq U_F(p)$. 
	
	\textbf{When $\mathcal{P}(root)=L$:} Pick node $v$ such that
	\begin{equation}\label{eq: pick h}
	    v\in \mathop{\arg\max}_{v'\in\mathrm{Supp}(p,root)}U_F(p)|_{{v'}}.
	\end{equation}
	By \Cref{thm: cheat in behavioral commit} (1)(c),  $p|_{v}$ is inducible at $T_v$. 
	By the inductive hypothesis, there exists a Y-shape inducible distribution $p^*|_{v}\in \Delta(Z_v)$ such that $U_F(p^*)|_{v}\geq U_F(p)|_{v}$ and $U_L(p^*)|_{v}\geq U_L(p)|_{v}$. We can simply extend $p^*|_{v}$ to $T_{root}$ by letting $p^*(z)=0$ for each $z\in Z\setminus Z_v$.
	
	One can verify that $p^*\in\Delta(Z)$ satisfies all three conditions of the case $\mathcal{P}(root)=L$ of \Cref{thm: cheat in behavioral commit}, which means $p^*$ is inducible. Furthermore, we have $U_F(p^*)\geq U_F(p)$ and $U_L(p^*)\geq U_L(p)$.
	
	This completes the proof.
\end{proof}

\subsection{Proof of \Cref{corollary: non-recursive characterization of y-shape inducibility}}

In the following proofs, given $z_1,z_2\in Z$ and $\alpha \in [0,1]$, we use $p_{\alpha,z_1,z_2}$ to denote the 
distribution $p\in \Delta(Z)$ that $p(z_1)=\alpha$ and $p(z_2)=1-\alpha$. 

\corollarynonrecursivecharacterizationofyshapeinducibility*

\begin{proof}
First notice that for a Y-shape distribution $p$ that $|\operatorname{Supp}(p)|=2$ to be realizable, let $\operatorname{Supp}(p)=\{z_1,z_2\}$, then the least common ancestor of $z_1$ and $z_2$ must belong to $\mathcal{P}^{-1}(L)$, which we denote as $LCA$. Now we first consider \textbf{Case (1)}:

\textbf{The necessity:}
Suppose for the sake of contradiction, condition (a) is not satisfied: 

If $z_1$ and $z_2$ do not have common ancestors belonging to $\mathcal{P}^{-1}(F)$, then for any $\alpha \in [0,1]$, any 
strategy profile of $p_{\alpha, z_1,z_2}$ is feasible for the leader to choose, whatever the follower's payoff function is. W.L.O.G., let $U_L(z_1)<U_L(z_2)$, then the pure strategy profile leading to $z_2$ always gains the leader higher utility.  Thus $p$ will never correspond to an SSE under any follower's payoff function. 

Now suppose for any common ancestors $v\in \mathcal{P}^{-1}(F)$ and $w\in\child{v}$ of $z_1$ and $z_2$, 
\begin{equation*}
    U_L(p)=U_L(p)|_{v} < \min_{w^\prime\in\child{v}, w^\prime \ne w} M_L(w^\prime),
\end{equation*}
then by \Cref{thm: cheat in behavioral commit}, $p|_{u}$ has to be inducible at $T_u$ for  any node $u$ on the path from $root$ to $LCA$. When $u=LCA\in \mathcal{P}^{-1}(L)$, that $U_L(z_1)\ne U_L(z_2)$ contradicts to condition (1)(a) of \Cref{thm: cheat in behavioral commit}. 

Condition (b) follows from condition (1)(b) of \Cref{thm: cheat in behavioral commit}.

\textbf{The sufficiency: } By condition (2)(a) of \Cref{thm: cheat in behavioral commit}, $p|_{v}$ is inducible at $T_v$. Suppose $p|_{u}$ is inducible at $T_u$ for ancestor $u$ of $v$, we prove that for $v$'s ancestor $t$ that $u\in\child{t}$, $p|_{t}$ is inducible at $T_t$: the case when $t\in \mathcal{P}^{-1}(F)$ follows from condition (2)(b) of \Cref{thm: cheat in behavioral commit}; when $\mathcal{P}^{-1}(L)$, conditions (1)(a) are satisfied and (1)(b)(c) of \Cref{thm: cheat in behavioral commit} follows from condition (b). 

Then we consider \textbf{Case (2)}:

\textbf{The necessity} follows from condition (1)(b) of \Cref{thm: cheat in behavioral commit}. 

\textbf{The sufficiency: } since $p$ is realizable, $v$ either belongs to $\mathcal{P}^{-1}(L)$ (when $z_1\ne z_2$) or is $z_1$ (when $z_1 = z_2$, $|\operatorname{Supp}(p)|=1$). By checking conditions of \Cref{thm: cheat in behavioral commit}, $p|_{v}$ is inducible at $v$. Suppose  $p|_{u}$ is inducible at $T_u$ for ancestor $u$ of $v$, for $v$'s ancestor $t$, that $u\in\child{t}$, $p|_{t}$ is inducible at $T_t$ by \Cref{thm: cheat in behavioral commit}. 
\end{proof}

\section{Omitted Proofs from Section \ref{sec: higher utility than pure}}

\subsection{Proof of \Cref{proposition: higher utility}}

\propositionhigherutility*

\begin{proof}
Recalling that any leaf node $z$ is equivalent to a distribution $p$ on leaf nodes that puts all probability on that leaf node, it suffices to show that $P(\Pi_L\times \Pi_F) \subseteq P(\Delta_L\times \Pi_F)$. For $p \in P(\Pi_L\times \Pi_F)$, we show by induction on the game tree that $p$ also satisfies conditions in \Cref{thm: cheat in behavioral commit}. 

\textbf{Inductive Base: }When $root \in Z$, the only realizable distribution $p$ where $p(root)=1$ is both inducible with pure and behavioral commitment. 

\textbf{Inductive Step: }When $root \in H$, let $\operatorname{Supp}(p,root)=\{v\}$. We assume the inductive hypothesis holds in subtrees and consider any $p\in P(\Pi_L\times \Pi_F)$. 

\textbf{When $\mathcal{P}(root) = L$,} by the definition of $P(\Pi_L\times \Pi_F)$, condition (1)(a) and (b) of \Cref{thm: cheat in behavioral commit} are satisfied.  Then $U_L(p)|_{v} = U_L(p) \geq M_L(root) = \max_{v^\prime\in \child{root}}M_L(v^\prime)\geq M_L(v)$. Thus $p|_{v}$ is inducible at $T_v$ with behavioral commitment by the inductive hypothesis, which means condition (1)(c) is satisfied. So $p$ is inducible with behavioral commitment 
in game $(T_{root}, U_L,U_F)$. 

\textbf{When $\mathcal{P}(root)=F$, } $U_L(p) \geq M_L(root) = \min_{v^\prime \in \child{root}}M_L(v^\prime)$. 

If $v\in \argmin_{v^\prime \in \child{root}}M_L(v^\prime)$, then $U_L(p)|_{v}=U_L(p) \geq M_L(v)$, $p|_{v}$ is inducible with both pure and behavioral commitment by the inductive hypothesis in subgame $(T_v,U_L|_{v},U_F|_{v})$, satisfying condition (2)(b) of \Cref{thm: cheat in behavioral commit}. Otherwise, $U_L(p)|_{v} \geq \min_{v^\prime \in \child{root}, v^\prime \ne v}M_L(v^\prime)$, satisfying condition 1. Thus $p$ is inducible with behavioral commitment in game $(T_{root},U_L,U_F)$. 

\end{proof}

\section{Omitted Proofs from Section \ref{sec: strong inducibility}}

\subsection{Proof of \Cref{thm: strong inducibility under pure}}

\thmstronginducibilityunderpure*

\begin{proof}
The proof is based on structural induction.

\textbf{Inductive Base: }When $root \in Z$, the strategy set is empty and $root$ is strongly inducible. 

\textbf{Inductive Step: }When $root \in H$, we first prove \textbf{the necessities}.

\textbf{When $\mathcal{P}(root)=L$:} assume that there exists $v'\in\child{root}$ that $v'\neq v$ and $U_L(z)\leq M_L(v')$ for the sake of contradiction. Then the leader can always get at least $M_L(v')$ by committing to a pure strategy $\pi_L$ such that $\pi_L(root)=v'$. So $z$ cannot be the unique outcome of all SSEs of game $(T_{root}, U_L,\tilde{U}_F)$ for any $\tilde{U}_F$.

Now assume that $z$ is not strongly inducible at 
$T_v$. 
Then for any $\tilde{U}_F$, if the leader optimally commits to some pure strategy $\pi_L$ that $\pi_L(root)=v$, then $z$ cannot be the unique outcome of all induced SSEs. On the other hand, if the leader optimally commits to some pure strategy $\pi_L$ that $\pi_L(root)\neq v$, then $z$ cannot be led to by any induced SSE, leading to a contradiction. 

\textbf{When $\mathcal{P}(root)=F$:} 
We assume $U_L(z)\leq \min_{v^\prime \in \child{root}, v^\prime\ne v} M_L(v^\prime)$ and prove that $z$ is strongly inducible at 
$T_v$. 

We claim that to make $z$ the unique outcome of SSEs, $\tilde{U}_F$ must satisfy that $\forall (\pi_L,\pi_F\in \BR(\pi_L))\in(\Pi_L\times\Pi_F)$,  $\pi_F(root)=v$. Otherwise, 
the leader must have a commitment $\pi^\prime_L$ and a corresponding follower's best response $\pi^\prime_F\in\BR(\pi^\prime_L)$ such that $\pi^\prime_F(root)=v'\neq v$.
Then the leader can get at least $M_L(v^\prime) \geq U_L(z)$ by committing to $\pi^\prime_L$, 
contradicting to that $z$ is the unique outcome. 

Now we suppose $\tilde{U}_F$ satisfies that for any $(\pi_L,\pi_F)\in(\Pi_L\times\Pi_F)$ if $\pi_F\in\BR(\pi_L)$, then $\pi_F(root)=v$. This reduces the problem to the subtree $T_v$. So if $z$ is strongly inducible on $T_{root}$, it must be strongly inducible on $T_v$.

Now we prove \textbf{the sufficiencies}.
\textbf{When $\mathcal{P}(root)=L$:} Let $\tilde{U}_F|_{v}$ be a payoff function such that all SSEs of subgame $(T_{v},U_L|_{v},\tilde{U}_F|_{v})$ lead to $z$. We extend it to $T_{root}$ by setting $\tilde{U}_F(z)=-U_L(z)$ for all $z\in Z\setminus Z_v$. Then if the leader commits to a strategy $\pi_L^\prime$ that $\pi_L^\prime(root)=v^\prime \ne v$, the best he get is $M_L(v^\prime) < U_F(z)$. thus his optimal commitment is to choose $v$ at $root$. Thus all SSEs of game $(T_{root},U_L,\tilde{U}_F)$ correspond to $z$. 

\textbf{When $\mathcal{P}(root)=F$:}
If (2)(a) holds, suppose that $v^\prime \in\child{root},v^\prime \neq v$ satisfies $U_L(z)>M_L(v^\prime)$. Construct the follower's payoff function, $\tilde{U}_F$, as follows:
\begin{eqnarray}
\tilde{U}_F(z')=\left\{ \begin{array}{ll}
-M_L(v^\prime)+1 & {z^\prime=z};\\
-M_L(v^\prime)-1 & z^\prime\not\in Z_{v^\prime}\wedge z^\prime\neq z;\\
-U_L(z^\prime) & z^\prime\in Z_{v^\prime}.
\end{array}
\right.
\end{eqnarray}
Since $U_L(z)>M_L(w)$, the best choice for the leader is to commit to a pure strategy that can achieve $z$.
Thus all SSEs lead to $z$. 

If (2)(b) holds, let $\tilde{U}_F|_{v}$ be a payoff function such that all SSEs of subgame $(T_{v},U_L|_{v},\tilde{U}_F|_{v})$ lead to $z$. We simply extend it by setting $\tilde{U}_F(z)=\min_{z\in Z_v}-\tilde{U}_F|_{v}(z)-1$ on $Z \setminus Z_v$. Then the follower will always best-respond to choose $v$ at $root$, and thus all SSEs of game $(T_{root},U_L,\tilde{U}_F)$ lead to $z$. 

This finishes the proof.
\end{proof}

\subsection{Proof of \Cref{thm: strong inducibility under behavioral}}

\thmstronginducibilityunderbehavioral*

\begin{restatable}{lemma}{lemmacharacterizefeasibleunderbehavioral}
\label{lemma: characterzie feasible under behavioral}
Given game $(T_{root}, U_L,U_F)$ and realizable distribution $p\in\Delta(Z)$, there exists a feasible
strategy profile 
$(\delta_L,\pi_F \in \BR(\delta_L))$ of $p$
if and only if $U_L(p)|_{v} \geq M_F(v)$ for any $v \in \{v'\in \mathcal{P}^{-1}(F): p(v')>0\}$. 
\end{restatable}

\begin{proof}
\textbf{The sufficiency: } Consider the following commitment of the leader
\begin{equation*}
    \delta_L(v) = \left \{\begin{array}{ll}
         \sum_{w\in\child{v}} \frac{p(w)}{p(v)} w, & p(v) > 0 ; \\
         \argmin_{w\in\child{v}}M_F(w), &  \text{otherwise;}
    \end{array}\right.
\end{equation*}
and the following strategy of the follower
\begin{equation*}
    \pi_F(v) = \left \{\begin{array}{ll}
         \sum_{w\in\child{v}} \frac{p(w)}{p(v)} w, & p(v) > 0 ; \\
         \argmax_{w\in\child{v}}M_F(w), &  \text{otherwise.}
    \end{array}\right.
\end{equation*}
Notice that $(\delta_L,\pi_F)$ correspond to $p$, and since $p$ is realizable, $\pi_F$ is pure. We prove that for any $v\in \mathcal{P}^{-1}(F)$ that $p(v)>0$, $\max_{\pi^\prime_F\in \Pi_F} U_F(\delta_L,\pi^\prime_F)|_{v} = U_F(\delta_L,\pi_F)|_{v}=U_F(p)|_{v}$. For $w\in \child{v}$ that $p(w) = 0$, by \Cref{lemma: max-min strategy}, we have $\pi_F \in \BR(\delta_L)|_{w}$, and the follower get at most $M_F(w)$ at $T_w$. Since $U_F(p)|_{v} \geq M_F(v) = \max_{w^\prime \in\child{v}}M_F(w^\prime) \geq M_F(w) $, the follower's choosing $w$ that $p(w)>0$ will lead to the best utility for the follower at $v$. 

\textbf{The necessity:} Suppose, for the sake of contradiction, that there is a 
$(\delta_L,\pi_F\in \BR(\delta_L))$ 
of
$p$ and a node $v\in \mathcal{P}^{-1}(F)$ that $p(v)>0$ satisfying $U_F(p)|_{v}<M_F(v)$. Then the follower can construct a new strategy $\pi^*_F$, where $\pi^*_F(v)=w\in\argmax_{w^\prime\in\child{v}}M_F(w^\prime)$, $\pi^*_F\in \BR(\pi_L)|_{w}$ and $\pi^*_F(v^\prime) = \pi_F(v')$ for other nodes $v^\prime \ne v$.  By \Cref{eqn: max-min full},  $U_F(\delta_L,\pi^*_F)|_{v} \geq M_F(w) = M_F(v) > U_F(\delta_L,\pi_F)|_{v}$, contradicting to that $(\delta_L,\pi_F\in \BR(\delta_L))$ by \Cref{lemma: feasible}. 
\end{proof}

\begin{proof}[Proof of \Cref{thm: strong inducibility under behavioral}]
We prove by structural induction over the game tree.

\textbf{Inductive Base: }When $root \in Z$, the only realizable distribution $p$ where $p(root)=1$ is strongly inducible, and all the conditions are satisfied. 

\textbf{Inductive Step: }When $root \in H$, we first prove \textbf{the necessities}. 

\textbf{When $\mathcal{P}(root)=L$:} consider a strongly inducible distribution $p$, then $p$ is also inducible. Suppose $\operatorname{Supp}(p,root)=\{v_1,v_2\} \subseteq \child{root}$. By the conditions of inducibility $U_L(p)|_{{v_i}}=U_L(p)$ for $i\in [2]$. Extend $p|_{{v_i}}$ to $p_i$ on $T_{root}$ by 
$p_i(z)=0$ for $z\in Z\setminus Z_{v_i}$, 
then $U_L(p_i)=U_L(p)|_{{v_i}}=U_L(p)$. 

For any
$\tilde{U}_F$ that makes a 
$(\delta_L,\pi_F)$ of $p$ SSE and $i\in [2]$, consider an extension of $(\delta_L,\pi_F)|_{{v_i}}$ to $T_{root}$, $(\delta^\prime_L,\pi^\prime_F)$, that $(\delta^\prime_L,\pi^\prime_F)(root)=v_i$, $(\delta^\prime_L,\pi^\prime_F)|_{{v_i}}=(\delta_L,\pi_F)|_{{v_i}}$ and is defined arbitrarily on other nodes. Then $(\delta^\prime_L,\pi^\prime_F)$ corresponds to $p_i$, is also feasible by \Cref{lemma: feasible} and $U_L(\delta^\prime_L,\pi^\prime_F) = U_L(\delta_L,\pi_F)$. Thus $(\delta^\prime_L,\pi^\prime_F)$ is an SSE of game $(T_{root}, U_L,\tilde{U}_F)$, leading to a contradiction. 

Now let $\operatorname{Supp}(p,root)=\{v\}$ and suppose $U_L(p) = \max_{v^\prime \in\child{root}, v^\prime \ne v}M_L(v^\prime)$, then for any 
payoff function that makes a 
strategy profile of $p$ SSE, by \Cref{lemma: min value in behavioral commit}, the leader can always gain at least $\max_{v^\prime\in\child{root},v^\prime \ne v}M_L(v^\prime)$ by choosing at $root$ from $\argmax_{v^\prime \in \child{root}, v^\prime \ne v}M_L(v^\prime)$ and thus not all SSEs correspond to $p$. 

If $p|_{v}$ is not strongly inducible at 
$T_v$, 
then for any  $\tilde{U}_F|_{v}$ on $T_v$ that makes a 
strategy profile 
$(\delta_L,\pi_F)|_{v}$ of $p|_{v}$ SSE, there is always another $p^\prime|_{v}$, such that $p^\prime|_{v}\ne p|_{v}$,  and one of its 
strategy profile $(\delta^\prime_L,\pi^\prime_F)|_{v}$ is also an SSE.
Extend $(\delta^\prime_L,\pi^\prime_F)|_{v}$ to $(\delta^\prime_L,\pi^\prime_F)$ on $T_{root}$ by letting $(\delta^\prime_L,\pi^\prime_F)(root)=v$, then $(\delta^\prime_L,\pi^\prime_F)$ corresponds to $p^\prime$ that $p^\prime = p^\prime|_{v}$ on $T_{v}$ and $p^\prime(z)=0$ otherwise.
Then $(\delta^\prime_L,\pi^\prime_F)$ is also an SSE, leading to a contradiction. 

\textbf{When $\mathcal{P}(root)=F$: }
first suppose for a Y-shape strongly inducible distribution $p$, $\operatorname{Supp}(p)=\{z_1,z_2\}$, and 
$U_L(z_1)=U_L(z_2)$. Consider any 
$\tilde{U}_F$ that makes a 
strategy profile $(\delta_L,\pi_F)$ of $p$ an SSE, and W.L.O.G., let $\tilde{U}_F(z_1)\leq \tilde{U}_F(z_2)$. Then $\tilde{U}_F(p)|_{v} \leq \tilde{U}_F(z_2)$ for any $v\in H$ on the path from $root$ to $z_2$ and $U_L(p)= U_L(z_2)$. Since $(\delta_L,\pi_F)$ is an SSE and so is feasible, $z_2$ satisfies conditions in \Cref{lemma: characterzie feasible under behavioral}. There exists a 
strategy profile 
$(\delta^\prime_L,\pi^\prime_F \in \BR(\delta^\prime_L))$ 
yielding
$z_2$. Thus $(\delta^\prime_L,\pi^\prime_F)$ is also an SSE, leading to a contradiction. 

For a strongly inducible distribution $p$ that $U_L(p) \leq \min_{v^\prime\in\child{root},v^\prime\ne v} M_L(v^\prime)$, we show that $p|_{v}$ is strongly inducible at 
$T_v$. 
For any $\tilde{U}_F$ that makes a 
strategy profile $(\delta_L,\pi_F)$ of $p$ SSE in game $(T_{root},U_L,\tilde{U}_F)$ and any feasible strategy profile $(\delta^\prime_L,\pi^\prime_F\in \BR(\delta^\prime_L))$ in game $(T_{root},U_L,\tilde{U}_F)$, it must be that $\pi^\prime_F(root)=v$, since otherwise the leader can always gain at least $\min_{v^\prime\in\child{root},v^\prime\ne v}M_L(v^\prime) \geq U_L(\delta_L,\pi_F)$ by committing to $\delta^\prime_L$ and lead to a different distribution. Then the problem reduces to the strong inducibility on subgame $T_v$,
thus $p|_{v}$ is strongly inducible at 
$T_v$. 

Now we prove \textbf{the sufficiency}.
Specifically, we construct the follower's payoff function as in the proof of \Cref{thm: cheat in behavioral commit}. the only difference is that when $p|_{v}$ is strongly inducible at subgame $T_v$,
we utilize the corresponding $\tilde{U}_F|_{v}$ at $T_v$ that all SSEs of subgame $(T_v, U_L|_{v},\tilde{U}_F|_{v})$ correspond to $p|_{v}$. 

We show that all SSEs of game $(T_{root},U_L,\tilde{U}_F)$ correspond to $p$. 

\textbf{When $\mathcal{P}(root)=L$: } First, by \Cref{thm: cheat in behavioral commit}, there is an SSE $(\delta_L,\pi_F)$ of game $(T_{root},U_L,\tilde{U}_F)$ 
that leads to 
$p$. Consider any SSE $(\delta^\prime_L,\pi^\prime_F)$ of game $(T_{root},U_L,\tilde{U}_F)$, then $(\delta^\prime_L,\pi^\prime_F)(root)=v$, since choosing any other child node makes the leader get at most $\max_{v^\prime\in\child{root},v^\prime\ne v} M_L(root) < U_L(p)=U_L(\delta_L,\pi_F)$. Then since all SSEs of 
$(T_v,U_L|_{v},\tilde{U}_F|_{v})$ correspond to $p|_{v}$, all SSEs of game $(T_{root},U_L,\tilde{U}_F)$ correspond to $p$.

\textbf{When $\mathcal{P}(root)=F$: }
\textbf{Case 1: When $U_L(p)> \min_{v^\prime\in \child{root},v^\prime \ne v} M_L(v^\prime)$. }
Consider any other SSE $(\delta^\prime_L,\pi^\prime_F\in \BR(\delta^\prime_L))$. 
First if $\pi^\prime_F(root)=v_0$, then the most the leader can get is $m_0 < U_L(\delta_L,\pi_F)$. When $\pi^\prime_F(root)\ne v_0$, by the proof of \Cref{thm: cheat in behavioral commit}, to make $U_L(\delta^\prime_L,\pi^\prime_F) = U_L(\delta_L,\pi_F)$, we have $A(\delta^\prime_L,\pi^\prime_F)\subseteq A(\delta_L,\pi_F)$. 
Since the only distribution that has $A(\delta_L,\pi_F)$ as its support and make the leader's gain $U_L(\delta_L,\pi_F)$ is $p$, thus $(\delta^\prime_L,\pi^\prime_F)$ correspond to $p$.

\textbf{Case 2: When $p|_{v}$ is strongly inducible at 
$T_v$. 
}Since the follower will always choose $v$ at $root$, and all SSEs of 
$(T_v, U_L|_{v},\tilde{U}_F|_{v})$ correspond to $p|_{v}$, all SSEs of game $(T_{root}, U_L,\tilde{U}_F)$ correspond to $p$.

\end{proof}

\subsection{Corollaries of \Cref{thm: algorithm for strongly inducible under behavioral} Necessary for the Following Proofs} 

Actually, one can find out that the proofs of necessities in \Cref{thm: strong inducibility under behavioral} almost apply to arbitrary strongly inducible distributions, except one condition. 

\begin{corollary}
\label{coro: necessities of arbitrary strong inducibility}
If a realizable distribution $p\in \Delta(Z)$ is strongly inducible, then 
\begin{enumerate}
    \item if $\mathcal{P}(root)=L$, then $|\operatorname{Supp}(p,root)
        |=1$. Let $\operatorname{Supp}(p,root)
        =\{v\}$, then all of the following conditions are met:
    \begin{enumerate}
        \item $U_L(p)> \max_{v^\prime \in \child{root}, v^\prime\ne v} M_L(v^\prime)$;
        \item $p|_{v}$ is strongly inducible at $T_v$. 
    \end{enumerate}
    \item if $\mathcal{P}(root)=F$,
     let $\operatorname{Supp}(p,root)
        =\{v\}$, then at least one of following two conditions are met: 
    \begin{enumerate}
        \item $U_L(p)> \min_{v^\prime \in \child{root}, v^\prime\ne v} M_L(v^\prime)$;
        \item $p|_{v}$ is strongly inducible at $T_v$. 
    \end{enumerate}
\end{enumerate}

\end{corollary}

\begin{proof}

\textbf{When $\mathcal{P}(root)=L$:} 
Consider a strongly inducible distribution $p$, then $p$ is also inducible. Suppose $\operatorname{Supp}(p,root)=\{v_1,\dots,v_k\} \subseteq \child{root}$ and $k\geq 2$. By the conditions of inducibility $U_L(p)|_{{v_i}}=U_L(p)$ for $i\in [k]$. Extend $p|_{{v_i}}$ to $p_i$ on $T_{root}$ by letting $p_i(z) = p|_{{v_i}}(z)$ for $z\in Z_{v_i}$ and $p_i(z)=0$ otherwise, then $U_L(p_i)=U_L(p)|_{{v_i}}=U_L(p)$. 

For any
$\tilde{U}_F$ that makes a 
strategy profile $(\delta_L,\pi_F)$ of $p$ SSE and $i\in [k]$, consider an extension of $(\delta_L,\pi_F)|_{{v_i}}$ to $T_{root}$, $(\delta^\prime_L,\pi^\prime_F)$, that $(\delta^\prime_L,\pi^\prime_F)(root)=v_i$, $(\delta^\prime_L,\pi^\prime_F)|_{{v_i}}=(\delta_L,\pi_F)|_{{v_i}}$ and is defined arbitrarily on other nodes. Then $(\delta^\prime_L,\pi^\prime_F)$ corresponds to $p_i$, is also feasible by \Cref{lemma: feasible} and $U_L(\delta^\prime_L,\pi^\prime_F) = U_L(\delta_L,\pi_F)$. Thus $(\delta^\prime_L,\pi^\prime_F)$ is an SSE of game $(T_{root}, U_L,\tilde{U}_F)$, leading to a contradiction. 

Now let $\operatorname{Supp}(p,root)=\{v\}$ and suppose $U_L(p) = \max_{v^\prime \in\child{root}, v^\prime \ne v}M_L(v^\prime)$, then for any 
payoff function that makes a 
strategy profile of $p$ SSE, by \Cref{lemma: min value in behavioral commit}, the leader can always gain at least $\max_{v^\prime\in\child{root},v^\prime \ne v}M_L(v^\prime)$ by choosing at $root$ from $\argmax_{v^\prime \in \child{root}, v^\prime \ne v}M_L(v^\prime)$ and thus not all SSEs correspond to $p$. 

If $p|_{v}$ is not strongly inducible at 
$T_v$, 
then for any  $\tilde{U}_F|_{v}$ on $T_v$ that makes a 
strategy profile $(\delta_L,\pi_F)|_{v}$ of $p|_{v}$ SSE, there is always another $p^\prime|_{v}$, such that $p^\prime|_{v}\ne p|_{v}$,  and one of its 
strategy profile $(\delta^\prime_L,\pi^\prime_F)|_{v}$ is also an SSE. Extend $(\delta^\prime_L,\pi^\prime_F)|_{v}$ to $(\delta^\prime_L,\pi^\prime_F)$ on $T_{root}$ by letting $(\delta^\prime_L,\pi^\prime_F)(root)=v$, then $(\delta^\prime_L,\pi^\prime_F)$ corresponds to $p^\prime$ that $p^\prime = p^\prime|_{v}$ on $T_{v}$ and $p^\prime(z)=0$ otherwise. Then for any  $\tilde{U}_F$ that makes $(\delta_L,\pi_F)$ an SSE, $(\delta^\prime_L,\pi^\prime_F)$ is also an SSE, leading to a contradiction. 

\textbf{When $\mathcal{P}(root)=F$: }
First suppose for a Y-shape strongly inducible distribution $p$, $\operatorname{Supp}(p)=\{z_1,z_2\}$, and 
$U_L(z_1)=U_L(z_2)$. Consider any 
$\tilde{U}_F$ that makes a 
strategy profile $(\delta_L,\pi_F)$ of $p$ an SSE, and W.L.O.G., let $\tilde{U}_F(z_1)\leq \tilde{U}_F(z_2)$. Then $\tilde{U}_F(p)|_{v} \leq \tilde{U}_F(z_2)$ for any $v\in H$ on the path from $root$ to $z_2$ and $U_L(p)= U_L(z_2)$. Since $(\delta_L,\pi_F)$ is an SSE and so is feasible, $z_2$ satisfies conditions in \Cref{lemma: characterzie feasible under behavioral} , and there exists a 
$(\delta^\prime_L,\pi^\prime_F \in \BR(\delta^\prime_L))$ 
leading to
$z_2$. Thus $(\delta^\prime_L,\pi^\prime_F)$ is also an SSE, leading to a contradiction. 

For a strongly inducible distribution $p$ that $U_L(p) \leq \min_{v^\prime\in\child{root},v^\prime\ne v} M_L(v^\prime)$, we show that $p|_{v}$ is strongly inducible at 
$T_v$. 
For any $\tilde{U}_F$ that makes a 
strategy profile $(\delta_L,\pi_F)$ of $p$ SSE in game $(T_{root},U_L,\tilde{U}_F)$ and any feasible strategy profile $(\delta^\prime_L,\pi^\prime_F\in \BR(\delta^\prime_L))$ in game $(T_{root},U_L,\tilde{U}_F)$, it must be that $\pi^\prime_F(root)=v$, since otherwise the leader can always gain at least  $\min_{v^\prime\in\child{root},v^\prime\ne v} M_L(v^\prime) \geq U_L(\delta_L,\pi_F)$ by committing to $\delta^\prime_L$ and result in a different distribution. Then the problem reduces to the strong inducibility on subgame
$T_v$, thus $p|_{v}$ is strongly inducible at $T_v$. 

\end{proof}

And to facilitate the following proofs, we state the non-recursive characterization for Y-shape strongly inducible distributions. 

\begin{restatable}{corollary}{corollarynonrecursivecharacterizationofyshapestronginducibility}
\label{corollary: non-recursive characterization of y-shape strong inducibility}
A  ``Y-shape'' realizable distribution $p$ is strongly inducible if and only if 
\begin{enumerate}
    \item if $|\operatorname{Supp}(p)|=2$, let $\operatorname{Supp}(p)=\{z_1,z_2\}$, then 
    \begin{enumerate}
        \item $U_L(z_1)\ne U_L(z_2)$; 
        \item $z_1$ and $z_2$ has common ancestors $v\in \mathcal{P}^{-1}(F)$ and $w\in \child{v}$ such that 
        \begin{equation*}
            U_L(p) >  \min_{w'\in\child{v},w'\neq w}M_L(w');
        \end{equation*}
        \item $U_L(p) > M_L(u)$ for $u\in \mathcal{P}^{-1}(L)$ on the path from $root$ to $v$;
    \end{enumerate}
    \item if $|\operatorname{Supp}(p)|=1$, let $\operatorname{Supp}(p)=\{z\}$, then
    \begin{equation*}
    U_L(p) > \max_{v \in \child{u}, v\ne v(u,z)}M_L(v)
    \end{equation*}
    for $u\in \mathcal{P}^{-1}(L)$ and $v(u,z) \in \child{u}$ on the path from $root$ to $z$.  
\end{enumerate}
\end{restatable}

\begin{proof}

First notice that for a Y-shape distribution $p$ to be realizable and $|\operatorname{Supp}(p)|=2$, let $\operatorname{Supp}(p)=\{z_1,z_2\}$, then the least common ancestor of $z_1$ and $z_2$ must belong to $\mathcal{P}^{-1}(L)$, which we denote as $LCA$. Now we first consider \textbf{Case (1)}:

\textbf{The necessity:}
Condition (a) follows from \Cref{thm: strong inducibility under behavioral}. 
Suppose for the sake of contradiction, condition (b) is not satisfied: by the proof of \Cref{corollary: non-recursive characterization of y-shape inducibility}, $z_1$ and $z_2$ must have at least one common ancestor in $\mathcal{P}^{-1}(F)$. Now suppose for any common ancestors $v\in \mathcal{P}^{-1}(F)$ and $w\in\child{v}$ of $z_1$ and $z_2$, 
\begin{equation*}
    U_L(p) \leq \min_{w^\prime\in\child{v}, w^\prime \ne w} M_L(w^\prime),
\end{equation*}
then by the conditions in \Cref{thm: strong inducibility under behavioral}, $p|_{u}$ has to be strongly inducible at  $T_u$ for any $u$ on the path from $root$ to $LCA$. When $u=LCA\in \mathcal{P}^{-1}(L)$,  that $|\operatorname{Supp}(p|_{{LCA}})|=2$  leads to a contradiction. 

Condition (c) follows from condition (1)(b) of \Cref{thm: strong inducibility under behavioral}.

\textbf{The sufficiency: } By condition (2)(a) of \Cref{thm: strong inducibility under behavioral}, $p|_{v}$ is strongly inducible at $T_v$. Suppose $p|_{u}$ is strongly inducible at $T_u$ for ancestor $u$ of $v$, we prove that for $v$'s ancestor $t$, that $u\in\child{t}$, $p|_{t}$ is strongly inducible at $T_t$: the case when $t\in \mathcal{P}^{-1}(F)$ follows from condition (2)(b) of \Cref{thm: strong inducibility under behavioral}; when $t\in\mathcal{P}^{-1}(L)$, conditions (1)(b) are satisfied and (1)(a) of \Cref{thm: strong inducibility under behavioral} follows from condition (c). 

Then we consider \textbf{Case (2)}: \textbf{The necessity} follows from condition (1)(b) of \Cref{thm: strong inducibility under behavioral}. 

\textbf{The sufficiency: } By checking conditions of \Cref{thm: strong inducibility under behavioral}, $p|_{z}$ is strongly inducible at $T_z$. Suppose  $p|_{u}$ is strongly inducible at $T_u$ for ancestor $u$ of $v$, for $v$'s ancestor $t$, that $u\in\child{t}$, $p|_{t}$ is strongly inducible at $T_t$ by \Cref{thm: strong inducibility under behavioral}. 
\end{proof}

Finally, we show that for any strongly inducible distribution $p\in \Delta(Z)$, there always exists a dominant Y-shape strongly inducible distribution $p^*$. 

\begin{corollary}
\label{coro: simple strongly inducible distribution}
For any extensive-form game $(T_{root},U_L,U_F)$ and any strongly inducible distribution $p$, there exists a Y-shape strongly inducible distribution $p^*\in \Delta(Z)$, such that $U_F(p^*)\geq U_F(p)$
 and $U_L(p^*) \geq U_L(p)$. 
\end{corollary}

\begin{proof}
We show by structural induction over the game tree. 

\textbf{Inductive Base: }When $root\in Z$, the only distribution $p$ that $p(root)=1$ is strongly inducible and is Y-shape. 

\textbf{Inductive Step: }When $root \in H$, suppose the inductive hypothesis holds. 

\textbf{When $\mathcal{P}(root)=L$: }Since $p$ satisfies the conditions in \Cref{coro: necessities of arbitrary strong inducibility}, let $\operatorname{Supp}(p,root)=\{v\}$, then $p|_{v}$ is strongly inducible at 
$T_v$. 
By the inductive hypothesis, there is a Y-shape strongly inducible distribution $p^*|_{v}$ on $T_{v}$ satisfies $U_L(p^*)|_{v} \geq U_L(p)|_{v}$, and $U_F(p^*)|_{v} \geq U_F(p)|_{v}$. Let $p^*$ be its extension to $T_{root}$ that $p^*(z)=0$ for $z\in Z\setminus Z_v$, then 
\begin{equation*}
    U_L(p^*) \geq U_L(p) > \max_{v^\prime\in\child{root},v^\prime\ne v}M_L(v^\prime). 
\end{equation*}
Thus $p^*$ is strongly inducible and $U_F(p^*) \geq U_F(p)$, $U_L(p^*) \geq U_L(p)$.

\textbf{When $\mathcal{P}(root)=F$: }
\textbf{Case 1: When $U_L(p)> \min_{v^\prime\in \child{root},v^\prime \ne v} M_L(v^\prime)$. } If $p$ is not Y-shape, then $|\operatorname{Supp}(p)|>2$, and $U_F(p)$ is a mixture of utilities of more than two leaf nodes. There exists a mixture of two leaf nodes $z_1$ and $z_2$ in the set $\{z|z\in Z, p(z)>0\}$, say $p^*$, satisfying $U_L(p^*)\geq U_L(p)> \min_{v^\prime\in \child{root},v^\prime \ne v} M_L(v^\prime)$ and $U_F(p^*)\geq U_F(p)$,. 

If $U_L(z_1) \ne U_L(z_2)$, then $p^*$ is strongly inducible. Else if $U_L(z_1)=U_L(z_2) \geq U_L(p_1)$, suppose $U_F(z_1)\leq U_F(z_2)$, then distribution $p^{**}$ that puts all probability on $z_2$ satisfies $U_F(p^{**})=U_F(z^*) \geq U_F(p)$. $p^{**}$ is Y-shape,  satisfies all conditions in \Cref{thm: strong inducibility under behavioral} and thus is strongly inducible.

\textbf{Case 2: When $p|_{v}$ is strongly inducible at $T_v$. 
} By the inductive hypothesis, there exists a  Y-shape strongly inducible distribution $p^*|_{v}\in \Delta(Z_v)$ satisfying $U_L(p^*)|_{v} \geq U_L(p)|_{v}$ and $U_F(p^*)|_{v} \geq U_F(p)|_{v}$. Let $p^*$ be its extension to $T_{root}$ that $p^*(z)=0$ for $z\in Z\setminus Z_v$, then $p^*$ is Y-shape and strongly inducible.

\end{proof}

\subsection{Proof of \Cref{thm: algorithm for strongly inducible under pure}}

\thmalgorithmforstronglyinducibleunderpure*

\begin{proof}
Note that we can enumerate each leaf node to check if it satisfies the conditions in \Cref{thm: strong inducibility under pure}. If there is no leaf node satisfying such conditions, then we conclude that there does not exist a strongly inducible leaf node. Otherwise, pick the strongly inducible leaf node with maximal follower's utility and construct a 
follower's payoff function to induce it by \Cref{thm: strong inducibility under pure}. This process can be solved in $O((|H|+|Z|)^2)$ time.
\end{proof}

\subsection{Proof of \Cref{thm: algorithm for strongly inducible under behavioral}}

\thmalgorithmforstronglyinducibleunderbehavioral*

\begin{proof}
We first use a lemma to prove that if $SP(\Delta_L\times\Pi_F)\ne \emptyset$, either there exists an optimal strongly inducible distribution, or for any $\epsilon$, an $\epsilon-$optimal one exists

\begin{lemma}\label{lemma: existence of near optimal}
For game $(T_{root}, U_L,U_F)$, if $SP(\Delta_L\times\Pi_F)\ne \emptyset$, then either
\begin{enumerate}
    \item there exists a distribution $p^*\in SP(\Delta_L\times\Pi_F)$, such that 
\begin{equation*}
    U_F(p^*)=\sup_{p\in SP(\Delta_L\times \Pi_F)}U_F(p)
\end{equation*} or
    \item  for any $\epsilon>0$, there exists $p(\epsilon)\in SP(\Delta_L\times\Pi_F)$, such that 
\begin{equation*}
U_F(p(\epsilon)) \geq \sup_{p\in SP(\Delta_L\times \Pi_F)}U_F(p)-\epsilon
\end{equation*}
\end{enumerate}
\end{lemma}

\begin{proof}[Proof of \Cref{lemma: existence of near optimal}]
By \Cref{coro: simple strongly inducible distribution}, it suffices to consider all the Y-shape strongly inducible distributions, which we denote as $SYP(\Delta_L\times\Pi_F)$. 

Note that, to prove the lemma, it suffices to prove that there exists a finite number of open intervals $\{(a_i,b_i)\}_{i\in [k]}$ and a finite number of points $\{r_i\}_{i\in [s]}$, such that 
$U_F(SYP(\Delta_L\times\Pi_F))=\left(\cup_{i\in [k]} (a_i,b_i)\right) \cup \left (\cup_{i\in[s]} r_i\right)$. 
For each $z\in Z$, define 
\begin{equation*}
    M(z) \coloneqq \underset{{\substack{u\in \mathcal{P}^{-1}(L)\text{~and~} v(u,z)\in \child{u}\\ \text{~on the path from~} root \text{~to~} z}}}{\max}\max_{v\in \child{u},v\ne v(u,z)}M_L(v),
\end{equation*} then by \Cref{corollary: non-recursive characterization of y-shape strong inducibility}, $SYP_1\coloneqq \{z\in Z: U_L(z) > M(z)\}$ consists of all strongly inducible distributions that has support size of $1$.  

Consider $(z_1,z_2) \in E(root)$, that is, the least common ancestor of $z_1,z_2 \in Z$ belongs to $\mathcal{P}^{-1}(L)$. And we consider those that $U_L(z_1)\ne U_L(z_2)$. For any of their common ancestors $v\in \mathcal{P}^{-1}(F)$ and $w\in \child{v}$, define
\begin{equation*}
    M(v,(z_1,z_2)) \coloneqq \max\{\underset{{\substack{u\in \mathcal{P}^{-1}(L)\\ \text{~on the path from~} root \text{~to~} v}}}{\max}M_L(u), \min_{w^\prime\in\child{v},w^\prime\ne w}M_L(w^\prime)\}
\end{equation*}
denote all of their common ancestors that belong to $\mathcal{P}^{-1}(F)$ as $CAF(z_1,z_2)$, and define 
\begin{equation*}
    M(z_1,z_2)\coloneqq \min_{v\in CAF(z_1,z_2)} M(v,(z_1,z_2)),
\end{equation*}
Given $z_1,z_2\in Z$ and $\alpha \in [0,1]$, recall that $p_{\alpha,z_1,z_2}$ denotes the 
distribution $p\in \Delta(Z)$ that $p(z_1)=\alpha$ and $p(z_2)=1-\alpha$. Then $SYP_{(z_1,z_2)}\coloneqq \{p_{\alpha,z_1,z_2}: \alpha\in (0,1), U_L(p_{\alpha,z_1,z_2})>M(z_1,z_2) \}$ denotes all the strongly inducible distributions that has support $\{z_1,z_2\}$.  $U_F(SYP_{(z_1,z_2)})$ is an open interval. 

Since $E(root)$ and $Z$ are finite sets, and \begin{equation*}
    U_F(SYP(\Delta_L\times \Pi_F)) = \left(\cup_{(z\in SYP_1)}U_F(z)\right) \cup \left(\cup_{(z_1,z_2)\in E(root)}U_F(SYP_{(z_1,z_2)})\right),
\end{equation*} 
$U_F(SYP(\Delta_L\times\Pi_F))$ is a union of a finite number of  of open intervals and a finite set. This finishes the proof. 

\end{proof}

Now we finish the proof of \Cref{thm: algorithm for strongly inducible under behavioral}.

The algorithm first calculates the sets $SYP_1$ and $SYP_{(z_1,z_2)}$ for each $(z_1,z_2)\in E(root)$ as in the proof of \Cref{lemma: existence of near optimal}.

Let $C\coloneqq \sup\left(\cup_{(z\in SYP_1)}U_F(z)\right) \cup \left(\cup_{(z_1,z_2)\in E(root)}U_F(SYP_{(z_1,z_2)})\right)$.

\textbf{Case 1:} If $C\in \left(\cup_{(z\in SYP_1)}U_F(z)\right) \cup \left(\cup_{(z_1,z_2)\in E(root)}U_F(SYP_{(z_1,z_2)})\right)$, then we enumerate $z\in SYP_1$ to check if $U_F(z)=C$ and enumerate $(z_1,z_2)\in E(root)$ to check if $C\in U_F(SYP_{(z_1,z_2)})$. Since $C\in \left(\cup_{(z\in SYP_1)}U_F(z)\right) \cup \left(\cup_{(z_1,z_2)\in E(root)}U_F(SYP_{(z_1,z_2)})\right)$, we must find a Y-shape 
distribution $p$ such that $U_F(p)=C$ and $p$ is strongly inducible.

\textbf{Case 2:} Otherwise, suppose that we are given an $\epsilon>0$. 
Then we enumerate $z\in SYP_1$ to check if $C-\epsilon\leq U_F(z)$ and enumerate $(z_1,z_2)\in E(root)$ to check if $C-\epsilon\leq \sup( U_F(SYP_{(z_1,z_2)}))$. Since $\left(\cup_{(z\in SYP_1)}U_F(z)\right) \cup \left(\cup_{(z_1,z_2)\in E(root)}U_F(SYP_{(z_1,z_2)})\right)$ is a union of a finite number of open intervals and a finite set, we must find a Y-shape 
distribution $p$ such that $U_F(p)\geq C-\epsilon$ and $p$ is strongly inducible.

Note that once we get a ``Y-shape'' 
distribution, we can construct a 
corresponding 
payoff function to induce it according to  the proof of \Cref{thm: cheat in behavioral commit} by \Cref{thm: strong inducibility under behavioral}.

This finishes the proof.
\end{proof}

\subsection{Proof of \Cref{proposition: higher utility under strongly inducibility}}

\propositionhigherutilityunderstronglyinducibility*

\begin{proof}
Same as in the proof of \cref{proposition: higher utility}, it suffices to show that $SP(\Pi_L\times \Pi_F) \subseteq SP(\Delta_L\times \Pi_F)$. Denote the set of all Y-shape strongly inducible distributions as $SYP(\Delta_L\times\Pi_F)$, since $p\in SP(\Pi_L\times \Pi_F)$ is Y-shape, we show by induction that $p$ satisfies conditions in \Cref{thm: strong inducibility under behavioral}. 

\textbf{Inductive Base: }When $root \in Z$, the only realizable distribution $p$ where $p(root)=1$ is both strongly inducible with pure and behavioral commitment. 

\textbf{Inductive Step: }When $root \in H$, we assume the inductive hypothesis holds in subtrees and consider any $p\in SP(\Pi_L\times \Pi_F)$. 

\textbf{When $\mathcal{P}(root) = L$,} by the definition of $SP(\Pi_L\times \Pi_F)$, $|\operatorname{Supp}(p,root)|=1$, and condition (1)(a) of \Cref{thm: strong inducibility under behavioral} is satisfied.  Let $\operatorname{Supp}(p,root)=\{v\}$, since $p_{v}$ is strongly inducible with pure commitment at subgame
$T_v$, 
by the inductive hypothesis, it is also strongly inducible with behavioral commitment. Thus condition (1)(b) is satisfied. 
$p$ is strongly inducible with behavioral commitment at $T_{root}$. 

\textbf{When $\mathcal{P}(root)=F$,} let $\operatorname{Supp}(p,root)=\{v\}$. 

If $U_L(p) > \min_{v^\prime \in \child{root},v^\prime \ne v}M_L(v^\prime)$, then $p$ satisfies condition (2)(a) of \Cref{thm: strong inducibility under behavioral} and is strongly inducible with behavioral commitment. 

If $p|_{v}$ is strongly inducible with pure commitment at $T_v$,
then  by the inductive hypothesis, it is also strongly inducible with behavioral commitment. Thus condition (2)(b) is satisfied. 
$p$ is strongly inducible with behavioral commitment at $T_{root}$. 

\end{proof}

\subsection{Proof of \Cref{lemma: equivalence of two near optimality}}

\lemmaequivalenceoftwonearoptimality*

\begin{proof}

\textbf{The neccesity} follows from that $SYP(\Delta_L\times\Pi_F)\subseteq SP(\Delta_L\times\Pi_F)$. As for the \textbf{the sufficiency}, 
denote an optimal inducible distribution as $p^*$ and suppose for any $\epsilon\geq 0$, there exists a strongly inducible distribution $p_1$ that $U_F(p_1)\geq U_F(p^*)-\epsilon$. By \Cref{coro: simple strongly inducible distribution}, there exists a Y-shape strongly inducible distribution $p_2$ that $U_L(p_2)\geq U_L(p_1)$, and $U_F(p_2) \geq U_F(p_1)$, which yields $U_F(p_2) \geq U_F(p^*)-\epsilon$.

\end{proof}

\subsection{Proof of \Cref{thm: full characterization of near optimality}}

\thmfullcharacterizationofnearoptimality*

By analyzing the possible optimal Y-shape inducible distributions a game may possess, we first prove two important lemmas.  

\begin{restatable}{lemma}{lemmayshapetwosupport}
\label{lemma: y-shape 2 support}
If game $(T_{root},U_L,U_F)$ satisfies \Cref{condition 1},  then it satisfies property \ref{eqn: near optimality}.
\end{restatable}

\begin{proof}
W.L.O.G., suppose $p(z_1)=\alpha>0$ and $U_F(z_1)\leq U_F(z_2)$, thus $U_F(p)=\alpha U_F(z_1) + (1-\alpha) U_F(z_2)$ and $U_L(p) = \alpha U_F(z_1) + (1-\alpha)U_F(z_2)$. Since $U_L(z_1)\ne U_L(z_2)$, for any small enough $\epsilon > 0$, there exists a $\alpha^\prime>0$, and distribution $p^\prime$ such that $p(z_1)=\alpha^\prime$ and $p(z_2) = 1-\alpha^\prime$, satisfying $U_F(p^\prime) > U_F(p)-\epsilon$, and $U_L(p^\prime) > U_L(p)$. 

Furthermore, suppose $v\in \mathcal{P}^{-1}(F)$ and $w\in\child{v}$ are common ancestors of $z_1$ and $z_2$ that $U_L(p)\geq \min_{w\in\child{v},w^\prime \ne w} M_L(w^\prime)$, the existence of which is assured by \Cref{corollary: non-recursive characterization of y-shape strong inducibility}, 
then 
\begin{equation*}
    \begin{aligned}
    U_L(p^\prime) &> U_L(p) \geq \min_{w\in\child{v},w^\prime \ne w} M_L(w^\prime);\\
    U_L(p^\prime) &>U_L(p) \geq M_L(u) ~\text{for any } u~\text{on the path from}~root~\text{to}~v.
    \end{aligned}
\end{equation*}
Thus $p^\prime$ is strongly inducible and game $(T_{root},U_L,U_F)$ satisfies Property \ref{eqn: near optimality}. 
\end{proof}

\begin{restatable}{lemma}{lemmayshapeonesupport}
\label{lemma: y-shape 1 support}
If game $(T_{root},U_L,U_F)$ does not satisfy \Cref{condition 1}, then it satisfies property \ref{eqn: y-shape near optimality} if and only if it satisfies \Cref{condition 2}. 
\end{restatable}

Given $z_1,z_2\in Z$ and $\alpha \in [0,1]$, recall that $p_{\alpha,z_1,z_2}$ denotes the 
distribution $p\in \Delta(Z)$ that $p(z_1)=\alpha$ and $p(z_2)=1-\alpha$. 

\begin{proof}

\textbf{The necessity: }
Now consider the case for all optimal Y-shape inducible distributions $p^*$, $p^*$ is not strongly inducible, and either $|\operatorname{Supp}(p^*)|=1$ or $U_L(z_1) = U_L(z_2)$ (which also means $U_F(z_1)=U_F(z_2)$) when $|\operatorname{Supp}(p^*)|=2$ and $\operatorname{Supp}(p^*)=\{z_1,z_2\}$. First noticing that for property \ref{eqn: y-shape near optimality} to be satisfied, for small enough $\epsilon>0$, there must exist a Y-shape strongly inducible distribution $p$ with $|\operatorname{Supp}(p)|=2$ such that $U_F(p)\geq U_F(p^*)-\epsilon$: Let $\epsilon_0 = \min_{z\in Z, U_F(z)\ne U_F(p^*)}|U_F(p^*)-U_F(z)|$. Since leaf nodes with utility $U_F(p^*)$ is not strongly inducible, for $\epsilon < \epsilon_0$, there does not exists a strongly inducible leaf node with utility at least $U_F(p^*)-\epsilon$. 

For any $\epsilon < \epsilon_0$, consider an aformentioned Y-shape strongly inducible distribution $p$ with $|\operatorname{Supp}(p)|=2$, and $\operatorname{Supp}(p)=\{z_1,z_2\}$,  then $U_L(z_1)\ne U_L(z_2)$, and $U_F(z_1)\ne U_F(z_2)$. W.L.O.G., assume $U_L(z_1)>U_L(z_2)$ and $U_F(z_1)<U_F(z_2)$, then $U_L(p^*)<U_L(z_1)$ and $U_F(p^*)\leq U_F(z_2)$. Again we argue that with small enough $\epsilon>0$, $z_2$ must be optimally inducible: otherwise 
denote the set of all pairs of leaf nodes $(z^\prime_1,z^\prime_2)$ as $I$, 
which satisfy that: neither node is optimally inducible, and
we define 
\begin{equation*}
    \epsilon_1 = U_F(p^*)- \max_{(z^\prime_1,z^\prime_2)\in I}\max_{\substack{\alpha \in [0,1] \\ p_{\alpha,z^\prime_1,z^\prime_2}\text{ is inducible}}} U_F(p_{\alpha, z^\prime_1,z^\prime_2}),
\end{equation*}
then $\epsilon_1>0$ and for $\epsilon <\epsilon_1$, there does not exist $(z^\prime_1,z^\prime_2)\in I$ and an inducible distribution $p_{\alpha, z^\prime_1,z^\prime_2}$ with utility at least $U_F(p^*)-\epsilon$.  Neither does a strongly inducible one. 

Now for any node $z_1$ that its least common ancestor with $z^*$, denoted as $LCA$ , belongs to $\mathcal{P}^{-1}(L)$ and $U_L(z_1)>U_L(z^*)$,  noticing that for distribution $p_{\alpha,z_1,z^*}$, where $\alpha \in (0,1)$ to be inducible, there must exist a node belonging to $\mathcal{P}^{-1}(F)$. If $U_L(z^*) < \min_{w\in\child{v}, w\ne w(v,z^*)}M_L(w)$ for all $v\in \mathcal{P}^{-1}(F)$ on the path from $root$ to $z^*$, then for small enough $\alpha>0$, $U_L(p_{\alpha,z_1,z^*})< \min_{w\in\child{v}, w\ne w(v,z^*)}M_L(w)$ and thus is not inducible, let alone strongly inducible. Thus there must exists a node $v$ on the path from $root$ to $LCA$, such that $U_L(z^*) \geq \min_{w\in\child{v}, w\ne w(v,z^*)}M_L(w)$. This finishes the proof of the necessities. 

\textbf{The sufficiency: } that $p^*$ is strongly inducible makes the game satisfies property \ref{eqn: near optimality} and thus \ref{eqn: y-shape near optimality}. Now suppose game $(T_{root},U_L,U_F)$ satisfies condition (2), then for any $\epsilon>0$, there exists $\alpha>0$, such that $U_F(p_{\alpha,z_1,z^*})>U_F(z^*)-\epsilon$, and $U_L(p_{\alpha,z_1,z^*})>U_L(z^*)$. Since $z^*$ is inducible, and $U_L(z^*) \geq \min_{w\in\child{v}, w\ne w(v,z^*)}M_L(w)$, thus $p_{\alpha,z_1,z^*}$ is strongly inducible.

\end{proof}

\begin{proof}[Proof of \Cref{thm: full characterization of near optimality}]
Since the optimal Y-shape inducible distributions either satisfies conditions in \Cref{lemma: y-shape 1 support} or in \Cref{lemma: y-shape 2 support}, and by \Cref{lemma: equivalence of two near optimality}, the theorem is proved. 
\end{proof}

\clearpage
\bibliographystyle{ACM-Reference-Format}

\end{document}